\documentclass[12pt]{article}

\usepackage[english]{babel}	
\usepackage[utf8x]{inputenc}		
\usepackage{enumerate}
\usepackage{textcomp}                
\usepackage{typearea}
\usepackage{pdfpages}
\usepackage[toc]{appendix}
\usepackage[]{todonotes}
\usepackage{subcaption}
\usepackage{verbatim}
\usepackage{comment}
\usepackage{mathtools}
\mathtoolsset{showonlyrefs}
\usepackage{lscape}

\usepackage{algorithm}
\usepackage{algpseudocode}
\usepackage{algorithmicx}
\usepackage[hidelinks]{hyperref}

\usepackage{amsfonts, amsmath, amssymb, amsthm, dsfont}
\usepackage{amsthm}			
\usepackage{thumbpdf}	
\usepackage{natbib}
\usepackage{booktabs}
\usepackage[noblocks]{authblk}
\usepackage{ulem}

\usepackage{pgf, tikz}
\usepackage{tikz-cd}		
\usepackage{bbm}

\theoremstyle{definition}
\newtheorem{definition}{Definition}[]

\newtheorem*{assumption*}{Assumption}

\newtheorem*{condition*}{Condition}

\theoremstyle{plain}
\newtheorem{thm}[definition]{Theorem}

\newtheorem{lemma}[definition]{Lemma}

\theoremstyle{remark}
\newtheorem{remark}{Remark}

\usepackage{graphicx}

\newcommand{\R}{\mathbb{R}}

\newcommand{\bX}{\boldsymbol{X}}

\newcommand{\bx}{\boldsymbol{x}}
\newcommand{\bU}{\boldsymbol{U}}
\newcommand{\balpha}{\boldsymbol{\alpha}}

\newcommand{\bzeta}{\boldsymbol{\zeta}}

\newcommand{\bgamma}{\boldsymbol{\gamma}}
\newcommand{\bbeta}{\boldsymbol{\beta}}
\newcommand{\PP}{\mathbb{P}}
\newcommand{\E}{\mathbb{E}}

\newcommand{\wconv}{\Rightarrow}

\newcommand{\Var}{\operatorname{Var}}
\newcommand{\Cov}{\operatorname{Cov}}
\allowdisplaybreaks

\includecomment{proofcomplete}

\title{Simultaneous Nonparametric Confidence Bands for Load-Sharing Systems}

\author[a]{Stefan Bedbur}
\author[b]{Johann K\"ohne}
\author[c]{Fabian Mies\thanks{corresponding author}}
\affil[a]{Institute of Statistics, RWTH Aachen University}
\affil[b]{Institute of Mathematical Stochastics, University of G\"ottingen}
\affil[c]{Delft Institute of Applied Mathematics, Delft University of Technology}
\begin{document}
	
	\maketitle

    \begin{abstract} 
\noindent Load-sharing systems arise in many different reliability applications, for instance, when modeling tensile strength of fibrous composites in textile industry or lifetimes of  redundant technical systems in engineering. Sequential order statistics serve as a flexible model for the ordered component failure times of such systems and allow  the residual lifetime distribution of the components to change after each component failure. In a proportional hazard rate setting, the model consists of some baseline distribution function and several model parameters describing successive adjustments of the hazard rates of the operating components. 
This work provides nonparametric confidence bands for the baseline distribution function, where the model parameters may be known or unknown. In case of known model parameters, we show how to construct exact confidence bands based on Kolmogorov-Smirnov type statistics, which are distribution-free with respect to the baseline distribution. If the model parameters are unknown, finite sample inference turns out to be infeasible, and asymptotic confidence bands for the baseline distribution function are derived.
As a technical tool, we extend the existing asymptotic theory of semiparametric estimators based on the profile-likelihood approach. \\[1ex]

\end{abstract}

\section{Introduction}
In many complex technical systems, load is shared among multiple components, and the failure of one part increases the stress on the remaining components. Examples include manufacturing systems \citep{arabzadeh_jamali_opportunistic_2022}, solar panels of satellites \citep{yang_reliability_2015}, power transmission lines \citep{jia_reliability_2023}, redundant micro-engine systems \citep{zhang_reliability_2020}, and composite materials \citep{SmiDan2014}.
To understand the reliability properties of these load-sharing systems, various statistical models have been proposed.
In this paper, we study sequential order statistics (SOSs), which have been introduced by \citet{Kam1995b,Kam1995a} as an extension of common order statistics and allow for a more flexible modelling of ordered component lifetimes in many reliability and engineering structures as, for instance, load-sharing systems, $k$-out-of-$n$ systems, or accelerated lifetime tests; see, e.g., \citet{CraKam2001b}, \citet{BalKamKat2012}, and \cite{BedKamKat2015}. Going beyond the possibilities of order statistics, SOSs enable for adjustments of the residual lifetime distribution of the  operating components upon each component failure. We illustrate the model by means of an example of a load-sharing system in textile industry; see \citet{SmiDan2014} and the references therein. Suppose that a steady tensile load is put on a bundle of $n$ identical textile fibers. Moreover, let $F_1,\dots,F_n$ denote cumulative distribution functions (cdfs) satisfying $F_1^{-1}(1-)\leq\dots\leq F^{-1}_n(1-)$ for technical reasons, and let $F_1$ be the initial lifetime distribution (time of tearing) of any of the $n$ fibers. After observing the first tear of some fiber at time $x_1$, say, the lifetime distribution of the intact $n-1$ fibers changes from $F_1$ to $(F_2(\cdot)-F_2(x_1))/(1-F_2(x_1))$ to model different stress conditions imposed on the remaining fibers. At time $x_2$, say, of the second fiber tear, the lifetime distribution of the still withstanding fibers changes to  $(F_3(\cdot)-F_3(x_2))/(1-F_3(x_2))$, and so on. Finally, the whole bundle is torn at time of the $n^{\text{th}}$ fiber tear.

While the general SOSs model is very flexible, restricted formulations allow for improved statistical inference.
In particular, a semiparametric SOSs model with proportional hazard rates is obtained by setting $F_j=1-(1-F)^{\alpha_j}$, $1\leq j\leq n$, for some absolutely continuous cdf $F$ with density function $f$ and positive numbers $\alpha_1,\dots,\alpha_n$. The hazard rate of $F_j$ is then given by $\lambda_{F_j}=\alpha_j\lambda_F$, where $\lambda_F$ denotes the hazard rate of $F$. In this context, $F$ is also referred to as baseline cdf, and $\alpha_1,\dots,\alpha_n$ are termed model parameters and describe the load-sharing characteristics of the system. For example, $\alpha_j< \alpha_{j+1}$ means that the failure of the $j^{\text{th}}$  component puts increased stress on the remaining components, thus leading to a higher failure rate.

In the distribution-theoretical sense, the proportional hazards SOSs model coincides with the model of generalized order statistics (GOSs) based on $F$, introduced by \citet{Kam1995b,Kam1995a}, in virtue of a bijective transformation of their  parameters. GOSs represent a unifying approach for various different models of ordered random variables allowing, among others, for a simultaneous probabilistic treatment. They are defined via the joint density  function of random variables $X^{(1)}\leq\dots\leq X^{(n)}$, involving a baseline cdf $F$ and positive parameters $\gamma_1,\dots,\gamma_n$. Respective choices of $\gamma_1,\dots,\gamma_n$ then result in densities corresponding to well-known models of ordered random variables, such as order statistics, record values, or progressively type-II censored order statistics. In particular, setting $\gamma_j=(n-j+1)\alpha_j$, $1\leq j\leq n$, gives the joint density function in the proportional hazards SOSs model as described above. Both models are therefore equal in distribution and only differ in terms of their parametrizations. For more details on GOSs, see \citet{Kam2016} and the references therein.

Distribution theory and statistical inference for the proportional hazards SOSs model have mostly been studied in the literature in fully parametric settings, where the baseline cdf $F$ is assumed to be known or to belong to some parametric class of cdfs; for an overview, see \cite{Cra2016} and, e.g., \citet{BedJohKam2019}, \citet{BedMie2022}, \citet{PesPolCriCra2023}, and \citet{PesCraCriPol2024} for some more recent works. The literature on nonparametric inference for SOSs is rather limited. Respective accounts are provided by \cite{kvam2005estimating}, \cite{beutner2008nonparametric, Beu2010b,beutner2010nonparametric}, and \cite{MieBed2019} with a focus on estimation and hypotheses testing in case that $F$ is unknown and $\alpha_1,\dots,\alpha_n$ are known or unknown.

In this paper, we derive semiparametric confidence bands for a nonparametrically specified baseline cdf $F$, where the model parameters $\alpha_1,\dots\alpha_n$ (or, equivalently, $\gamma_1,\dots,\gamma_n$ when using the GOSs parametrization) are assumed to be either known or unknown. In contrast to pointwise confidence intervals for $F(x)$, such a confidence band is a random set in the two-dimensional plane that contains the entire graph of $F$ with a specified probability. The proposed bands are based on the Nelson-Aalen type estimator $\widehat{F}$ for $F$ and the corresponding profile-likelihood estimator $\widehat{\bgamma}$ for the vector $\bgamma$ of load-sharing parameters, as introduced by \cite{kvam2005estimating}.
Building on the material in \cite{MieBed2019}, we first construct exact confidence bands for $F$ for finite (small) samples by making use of certain Kolmogorov-Smirnov type statistics with distributions independent of $F$.
In case of unknown model parameters, however, the construction turns out to be intractable, and we therefore pursue an asymptotic approach to address the problem.
While the limit theory for $\widehat{F}$ and $\widehat{\bgamma}$ has been put forward by \cite{kvam2005estimating}, we found the need to extend it in two relevant ways: (i) We develop the asymptotic theory for $\widehat{F}$ and $\widehat{\bgamma}$ based on all observations, whereas former results require some artificial type-I censoring of the data (using only information up to a fixed time point) for purely technical reasons. 
(ii) We fully specify the asymptotic covariance structure of $\sqrt{M}(\widehat{F}-F)$, which is known to converge weakly to a Gaussian process as the sample size $M$ tends to infinity.
The refined asymptotic theory is then utilized to construct asymptotic confidence bands for $F$ on the whole real line. 

As the proportional hazards SOSs model is used in various reliability applications, the proposed bands provide different insights, depending on the context of the experiment. In a progressively type-II censored lifetime experiment, where a fixed number of intact components is removed from the experiment after each component failure, the model parameters are completely determined by the censoring scheme and thus known, such that the bands of Sections \ref{ss:EBknown} and \ref{ss:ABknown} are applicable to assess the underlying lifetime distribution of the components. On the other hand, when describing the component lifetimes in some step-stress experiment in accelerated life testing, the model parameters correspond to different stress levels and are usually unknown. In this case, the bands of Section \ref{ss:ABunknown} apply and yield information about the component  lifetime distribution under normal operating conditions. In the following, we shall stay within  the interpretation in terms of load-sharing systems.

The remainder of this paper is structured as follows.
Upon a concise introduction to GOSs in Section \ref{s:gos}, we show how to obtain exact confidence bands for $F$ via test inversion in Section \ref{s:eb}. Section \ref{ss:ABknown} then describes the construction of asymptotic confidence bands for $F$ if  $\bgamma$ is known, and Section \ref{ss:ABunknown} presents respective bands in the practically most relevant situation of unknown load-sharing parameters. Here, we also provide the extended asymptotic theory for $\widehat{F}$ and $\widehat{\bgamma}$ (Theorem \ref{thm:jointclt}).
The methodology is briefly illustrated in Section \ref{s:data} by means of two data examples. 
All proofs and several auxiliary results are postponed to the Appendix for readability.


\section{Generalized Order Statistics}\label{s:gos}

Let $F$ denote some absolutely continuous cdf with corresponding density function $f$, and let $\gamma_1,\dots,\gamma_r$ denote positive numbers. Then the random variables $X^{(1)}\leq\dots\leq X^{(r)}$ are called GOSs based on $F$ and $\bgamma=(\gamma_1,\dots,\gamma_r)$ if their joint density function is given by
\begin{equation}\label{eq:densGOS}
f^{\bX}(\bx)\,=\,\left(\prod_{j=1}^r\gamma_j\right)\,\left(\prod_{j=1}^r (1-F(x_j))^{\gamma_j-\gamma_{j+1}-1} f(x_j)\right)
\end{equation}
for $\bx=(x_1,\dots,x_r)\in\mathcal{X}=\{(y_1,\dots,y_r)\in\mathbb{R}^r: F^{-1}(0+)<y_1<\dots<y_r<F^{-1}(1-)\}$; see \citet{Kam1995b,Kam1995a}. Here, $F^{-1}$ denotes the quantile function of $F$, and we set $\gamma_{r+1}=0$ for a closed representation. 

Several important models of ordered random variables are included in the model of GOSs in the distribution-theoretical sense by respective choices of $\bgamma$. For $\gamma_j=n-j+1$, $1\leq j\leq r$, formula (\ref{eq:densGOS}) is the joint density function of the first $r$ out of $n$ (common) order statistics based on $F$, while for $\gamma_j=1$, $1\leq j\leq r$, the joint density function of the first $r$ (common) record values based on $F$ is obtained. As pointed out in the introduction, setting 
\begin{equation*}
\gamma_j\,=\,(n-j+1)\,\alpha_j\,,\qquad 1\leq j\leq r\,,
\end{equation*}
where $\alpha_1,\dots,\alpha_n$ denote positive numbers, formula (\ref{eq:densGOS}) represents the joint density function of the first $r$ out of $n$ SOSs based on the cdfs $F_j=1-(1-F)^{\alpha_j}$ with hazard rates $\lambda_{F_j}=\alpha_j\lambda_F$ for $1\leq j\leq n$. Here, $\lambda_F=f/(1-F)$ denotes the hazard rate of $F$ with corresponding cumulative hazard rate $\Lambda_F=-\log(1-F)$. In what follows, we refer to $F$ as baseline distribution and to $\gamma_1,\dots,\gamma_r$ as load-sharing parameters, indicating their relation to the multipliers  $\alpha_1,\dots,\alpha_r$ of the baseline hazard rate and the use of GOSs as a model for load-sharing systems, here.

Throughout the paper, we assume that the baseline cdf $F$ has support $[0,\infty)$, which implies that $F$ is strictly increasing on $[0,\infty)$ and satisfies $F(0)=0$. The collection of all absolutely continuous cdfs with that property is denoted by $\mathcal{F}$. The GOSs model is then in one-to-one correspondence with the counting process
\begin{equation*}
N(t)\,=\,\sum_{j=1}^r \mathds{1}\{X^{(j)} \leq t\}\,,\quad t\geq0\,, 
\end{equation*}
with cumulative intensity
\begin{align*}
    L(t)\,=\,\int_0^t \gamma(s) \lambda_F(s)\,ds,\quad t\geq0\,,
  \end{align*}
  where
\begin{equation}\label{eq:gamma(s)}
    \gamma(s) \,=\, \sum_{j=1}^r \gamma_j \mathds{1}\{X^{(j-1)}<s \leq X^{(j)}\} \,=\, \gamma_{N(s)+1}\,,\quad s\geq0\,,
\end{equation}
and $X^{(0)}=0$; see \citet{kvam2005estimating}, \citet{beutner2008nonparametric}, and \citet{MieBed2019}. Here, $\gamma(s)$ denotes the load parameter $\gamma_j$ which is active at time $s$. Moreover, we use the notation $\mathds{1}A$ for the indicator function of set $A$.

\section{Exact Confidence Bands}\label{s:eb}

We assume to have $M$ independent vectors $\bX_i=(X_i^{(1)},\dots,X_i^{(r)})$, $1\leq i\leq M$, of GOSs based on $F\in\mathcal{F}$ and $\bgamma$ with density function (\ref{eq:densGOS}), each.  For $1\leq i\leq M$, let
\begin{equation*}
N_i(t)=\sum_{j=1}^r \mathds{1}{\{X_i^{(j)} \leq t\}}\,,\qquad t\geq0\,,
\end{equation*}
denote the counting process associated with $\bX_i$ having instantaneous load parameter
\begin{equation}\label{eq:gammaI}
\gamma^{(i)}(s)\,=\, \sum_{j=1}^r \gamma_j \mathds{1}\{X_i^{(j-1)}<s \leq X_i^{(j)}\}\,,\qquad s\geq0\,,
\end{equation}
where $X_i^{(0)}=0$. 

It is well-known that the transformed random vectors
\begin{equation}\label{eq:Ui}
    \bU_i=(U_i^{(1)},\dots,U_i^{(r)})=(F(X_i^{(1)}),\dots,F(X_i^{(r)})), \qquad1\leq i\leq M,
\end{equation}
are distributed as $M$ independent vectors of GOSs based on $G$ and $\bgamma$, where $G$ denotes the cdf of the standard uniform distribution with hazard rate $\lambda_G(u)=1/(1-u)$ and cumulative hazard rate $\Lambda_G(u)=-\log(1-u)$ for $u\in[0,1)$. In particular, we have the relation
\begin{equation}\label{eq:cumhaz}
\Lambda_F(t) = \Lambda_G(F(t))\,,\qquad t\geq0\,.
\end{equation}

\subsection{Known Model Parameters}\label{ss:EBknown}

For known $\bgamma$, exact confidence bands for (the graph of) $F$ can be constructed based on the Nelson-Aalen estimator 
\begin{equation*}
\widehat{F}_{\bgamma}(t)\,=\,\,1-\prod_{X_i^{(j)}\leq t} \left(1-Y_i^{(j)}\right)\,,\qquad t\geq0\,,
 \end{equation*}
for $F$ with
 \begin{equation*}
	Y_i^{(j)}\,=\,\left( \sum_{k=1}^M \gamma_{N_k(X_i^{(j)}-)+1} \right)^{-1}\,,\qquad 1\leq j\leq r\,,\qquad 1\leq i\leq M\,,
 \end{equation*}
 which is proposed in \citet{kvam2005estimating}. Since $F$ is strictly increasing on $[0,\infty)$, we obtain the useful identity
 \begin{equation}\label{eq:NAErel}
 \widehat{F}_{\bgamma}(t)\,=\,\widehat{G}_{\bgamma}(F(t))\,,\qquad t\geq0\,,
 \end{equation}
 where $\widehat{G}_{\bgamma}$ denotes the Nelson-Aalen estimator (of $G$) based on $\bU_1,\dots,\bU_M$; cf. formula (10) in \citet{MieBed2019}.

 \begin{thm}\label{thm:EBknown}
Let $q\in(0,1)$ and $H:\R^2\rightarrow\R$ be a continuous mapping. A confidence band with level $q$ for the graph $\text{gr}(F)=\{(t,F(t)):t>0\}$ of $F\in\mathcal{F}$ on $(0,\infty)$ is given by
\begin{equation*}
B_{\bgamma}\,=\,\{(x,y)\in(0,\infty)\times[0,1):\,H(\widehat{F}_{\bgamma}(x),y)\leq c_{\bgamma}(q)\}\,,
 \end{equation*}
 where $c_{\bgamma}(q)$ is the $q$-quantile of the statistic
 \begin{equation*}
K_{\bgamma}\,=\,\sup\limits_{u\in(0,1)} H(\widehat{G}_{\bgamma}(u),u)\,,
 \end{equation*}
 which is distribution-free with respect to $F$. 
\end{thm}

For $H(y,z)=|y-z|$, $B_{\bgamma}$ is a Kolmogorov-Smirnov type band with vertical width not larger than $2c_{\bgamma}(q)$ at every $x>0$. 
Alternatively, we may consider a function of the form $H(y,z)=h(y)|y-z|$ for some weighting function $h$, suitable choices of which will be discussed in Section \ref{s:ab} below.

\subsection{Unknown Model Parameters}\label{ss:EBunknown}

If $\bgamma$ is unknown, constructing an exact confidence band for $F\in\mathcal{F}$ is complicated. First, let us assume that $\gamma_1=n$ to guarantee that $F$ and $\gamma_2,\dots,\gamma_r$ are identifiable. As a starting point, we consider a statistic of type
\begin{equation}\label{eq:HpivotF}
T_{F,\bgamma}\,=\,\sup_{t>0}H(\widehat{F}_{\bgamma}(t),F(t))
\end{equation}
for some continuous function $H:\R^2\rightarrow\R$, which forms a pivotal quantity for $F$ in case of $\bgamma$ being known; see the proof of Theorem \ref{thm:EBknown} in the Appendix. In \citet{kvam2005estimating}, a profile-likelihood estimator $\widehat{\bgamma}$ for $\bgamma$ is proposed, which is shown to be distribution-free with respect to $F$ in \cite{MieBed2019}, Theorem 1. Inserting for $\bgamma$ in formula (\ref{eq:HpivotF}), the distribution of  $T_{F,\widehat{\bgamma}}$ would be free of $F$ as well but still depend on $\bgamma$. As the asymptotic covariance of $\widehat{\bgamma}$ derived in \citet{kvam2005estimating} suggests, there is no simple way to obtain a pivotal statistic from it. 

In the context of testing a simple null hypothesis $H_0: F=F_0$ for some given $F_0\in\mathcal{F}$, \cite{MieBed2019} remove the dependence (in distribution) of a respective test statistic on the unknown nuisance parameter $\bgamma$ by constructing a test with Neyman structure. Here, the basic idea is to use that, in case $F$ is known, the maximum likelihood estimator $\tilde{\bgamma}_F$, say, of $\bgamma$ is a sufficient statistic for $\bgamma$ and has a distribution not depending on $F$; see, e.g., \cite{CraKam1996}. Hence, the conditional distribution of $T_{F,\tilde{\bgamma}_F}$ given $\tilde{\bgamma}_F=\boldsymbol{\eta}$ and thus also its $q$-quantile  $c_q(\boldsymbol{\eta})$ do neither depend on $F$ nor on $\bgamma$.
The test statistic for $H_0:F=F_0$ is then defined as $T_{F_0,\tilde{\bgamma}_{F_0}}$, 
and, in case that $H_0$ is rejected for large values of $T_{F_0,\tilde{\bgamma}_{F_0}}$, the acceptance region of the test with level $1-q$ is given by
\begin{equation*}
    A(F_0)\,=\,\{T_{F_0,\tilde{\bgamma}_{F_0}}\leq c_q(\tilde{\bgamma}_{F_0})\}\,.
    \end{equation*}

Inverting the acceptance region now yields that
\begin{equation}\label{eq:Binv}
   B\,=\,\{F_0\in\mathcal{F}:\,T_{F_0,\tilde{\bgamma}_{F_0}}\leq c_q(\tilde{\bgamma}_{F_0})\}
    \end{equation}
forms a confidence set for $F\in\mathcal{F}$ with exact level $q$. Although Theorem 4  in \cite{MieBed2019} can be applied to obtain the quantiles of the conditional distribution via Monte Carlo methods, the implicit and complicated representation of the set in formula (\ref{eq:Binv}) seems to be too cumbersome for users, which gives rise to an asymptotic approach in a later section.

\section{Asymptotic Confidence Bands}\label{s:ab}

We now turn to the construction of asymptotic confidence bands for $F\in\mathcal{F}$. For $\bgamma$ being known, asymptotic bands for $F$ are developed in Section \ref{ss:ABknown} and seen to require less computational effort than the exact counter-parts in Section \ref{ss:EBknown}, thus constituting practical competitors if the number $M$ of samples is large. Section \ref{ss:ABunknown} then offers asymptotic bands for $F$ under the assumption that $\bgamma$ is unknown, in the case of which simple exact  procedures are not available; see Section \ref{ss:EBunknown}.

\subsection{Known Model Parameters}\label{ss:ABknown}

If the load-sharing parameters $\gamma_1,\dots,\gamma_r$ are known, the cumulative hazard rate $\Lambda_F$ may be estimated nonparametrically by the step function
\begin{equation}\label{eq:hatLa}
    \widehat{\Lambda}_F(t) \,=\,\int_0^t \frac{d\overline{N}(s)}{\overline{\gamma}(s)} \,,\qquad t\geq0\,,
    \end{equation}
    with
    \begin{equation*}
    \overline{\gamma}(s) \,=\,\frac{1}{M} \sum_{i=1}^M \gamma^{(i)}(s)\qquad\text{and}\qquad\overline{N}(s) \,=\, \frac{1}{M} \sum_{i=1}^M N_i(s)\,,\qquad s\geq0\,,
\end{equation*}
having jumps at $X_i^{(j)}$, $1\leq j\leq r$, $1\leq i\leq M$; see \citet[equation (5)]{kvam2005estimating} and \citet{beutner2010nonparametric, beutner2008nonparametric}.
As $M\rightarrow\infty$, \cite{kvam2005estimating} have established the functional weak convergence
\begin{align}
	\label{eqn:hazardlimit}\sqrt{M} (\widehat{\Lambda}_F(t)-\Lambda_F(t)) \;\wconv\; W_{\bgamma}(t) 
\end{align}
in the space $D[0,T]$ of cadlag functions on $[0,T]$, where $T$ is such that $\E(\gamma(T))>0$ and $W_{\bgamma}(t)$, $t\geq0$, is a zero-mean Gaussian process with independent increments and variance function 
\begin{equation*}
\Var(W_{\bgamma}(t))\,=\,\int_0^t \frac{\lambda_F(s)}{\E(\gamma(s))}\, ds\,,\qquad t\geq0\,.
\end{equation*}
The jump-sizes of $\widehat{\Lambda}_F(t)$ only depend on the ranks of the failure times $X_i^{(j)}$ , $1\leq j\leq r$, $1\leq i\leq M$, among all failure times, and we may write 
\begin{align}
	\label{eqn:shift}\widehat{\Lambda}_F(t)=\widehat{\Lambda}_G(F(t))\,,\qquad t\geq0\,,
\end{align}
where $\widehat{\Lambda}_G$ denotes the estimator of $\Lambda_G$ based on $\bU_1,\dots,\bU_M$, as defined in formula (\ref{eq:Ui}). As a consequence, a quantity of the form
\begin{eqnarray}
&&\sup_{t\geq0} \left|  H(\widehat{\Lambda}_F(t),F(t))-H(\Lambda_F(t),F(t)) \right| \notag\\
 &=&\, \sup_{t\geq0} \left|  H(\widehat{\Lambda}_G(F(t)),F(t))-H(\Lambda_G(F(t)),F(t)) \right| \notag\\
	&=&\, \sup_{u\in[0,1)} \left|  H(\widehat{\Lambda}_G(u),u)-H(\Lambda_G(u),u) \right|\label{eqn:distr-free-band}
    \end{eqnarray}
is distribution-free with respect to the baseline cdf, where $H$ is any appropriate function for which the supremum is finite. However, its distribution does depend on $\bgamma$ influencing the distribution of the ranks of $X_i^{(j)}$, $1\leq j\leq r$, $1\leq i\leq M$.

To construct an asymptotic confidence band for $F$, we therefore study the limit process in formula (\ref{eqn:hazardlimit}) for a standard uniform baseline distribution in more detail. The quantity $\E(\gamma(s))$ 
in the associated variance function
\begin{equation}\label{eq:tau}
	v_{\bgamma}(u)\,=\,\int_0^u \frac{\lambda_G(s)}{\E(\gamma(s))} ds\,,\qquad u\in[0,1]\,,
\end{equation}
of $W_{\bgamma}(u)$, $u\in[0,1]$, can be determined by means of Monte-Carlo simulations. In case of pairwise distinct model parameters, a useful representation can be derived analytically.

\begin{lemma}\label{la:EN(t)}
	\label{Lemma: PN_t=r}
	Let the baseline distribution be a standard uniform distribution,
 \begin{align}
	\label{GammaDistinct}
	\gamma_i\neq\gamma_j\,,\quad 1\leq i,j\leq r\,,\quad i\neq j\,,
\end{align}
and $\gamma(u)$, $u\in[0,1]$, be the corresponding instantaneous load parameter as defined  in formula (\ref{eq:gamma(s)}). Then, we have that
\begin{equation}\label{eq:E(ga)}
\E(\gamma(u))\,=\,\sum_{i=1}^{r} c_i(\gamma)\,(1-u)^{\gamma_i}\,,\qquad u\in[0,1]\,,
	\end{equation}
	where
	\begin{align}\label{eq:ci}
		c_i(\bgamma)\,=\,\sum_{k=i}^r \left(\prod_{j=1}^k \gamma_j\right) \left(\prod_{j=1, j\neq i}^{k}(\gamma_j-\gamma_i)\right)^{-1}\,,\qquad i\in\{1,\dots,r\}\,.
	\end{align}
\end{lemma}

In case that formula (\ref{GammaDistinct}) is true, Lemma \ref{la:EN(t)} removes the need for Monte-Carlo methods, and some numerical integration scheme may then be applied to obtain the value of the integral
\begin{equation}\label{eq:vgamma}
	v_{\bgamma}(u)\,=\,\int_0^u \frac{\lambda_G(s)}{\E(\gamma(s))} ds \,=\,\int_0^u \left(\sum_{i=1}^r c_i(\bgamma)\,(1-s)^{\gamma_i+1} \right)^{-1} ds
\end{equation}
 as a function of $u\in[0,1]$, which is crucial for determining the confidence band. In particular, the lemma applies to common order statistics and progressively type-II censored order statistics.
 \begin{remark}\label{rem:egamma}
     If we drop assumption \eqref{GammaDistinct} and allow for $\gamma_i=\gamma_j$ for some $1\leq i,j\leq r$ with $i\neq j$, Lemma \ref{Lemma: PN_t=r} is in this form no longer valid. However, it can be extended to this edge case, since the mapping $\bgamma \mapsto \E(\gamma(u))$ is continuous. In particular, $\E(\gamma(u))$ is still of the form in equation \eqref{eq:E(ga)} for appropriate coefficients $c_i(\bgamma)$, $1\leq i\leq r$, different from those in formula \eqref{eq:ci}.
 \end{remark}
 Now, recall that the functional weak convergence in formula (\ref{eqn:hazardlimit}) holds true on any compact interval $[0,T]$ with $\E(\gamma(T))>0$. In case of a standard uniform baseline cdf, the condition simplifies to $T<1$ by using formula (\ref{eq:gammaI}). This finding can be utilized to state the following asymptotic result.

\begin{thm}
	\label{Theorem: ConfBand}
 Let $q\in (0,1)$ and the function $g:(0,\infty)\rightarrow\R$ be such that $g\geq\varepsilon$ for some $\varepsilon>0$ and $\sqrt{x\log\log x}/g(x)\rightarrow0$ for $x\rightarrow\infty$. Moreover, let
  \begin{align*}
    \tilde{B}_{\bgamma}^{(M)}\,=\,\left\{(x,y)\in(0,\infty)\times(0,1):\,
    \frac{\sqrt{M}}{g(v_{\bgamma}(\widehat{F}(x)))}\frac{|\widehat{F}(x)-y|}{1-\widehat{F}(x)} \leq d(q) \right\} 
\end{align*}
with $\widehat{F}=1-\exp\{-\widehat\Lambda_F\}$  and $v_{\bgamma}$ as in formula (\ref{eq:tau}), where $d(q)$ denotes the $q$-quantile of the real-valued random variable
\begin{equation}
\sup_{z\in (0,\infty)} \frac{|W_*(z)|}{g(z)}
   \end{equation}
   for a standard Brownian motion $W_*$. Then $\tilde{B}_{\bgamma}^{(M)}$ is a confidence band for the graph of $F\in\mathcal{F}$ with asymptotic level $q$ in the sense that
 \begin{equation}\label{eqn:confidence-band}
     \forall\, R>0: \quad \lim_{M\to\infty} \PP \left(\{(t,F(t)):\, t\in (0, R]\}\,\subset\,\tilde{B}_{\bgamma}^{(M)} \right) \geq q. 
 \end{equation}
\end{thm}

As simple candidates for $g$ in Theorem \ref{Theorem: ConfBand}, we may consider functions of the form  $g(x)=x^a+b$, $x>0$, for some $a>1/2$ and $b>0$, or $g(x) = \sqrt{x \log(x+1)}+c$, $x>0$, for some $c>0$. In cases (a) and (c) of Figure \ref{fig:example-bands}, the asymptotic confidence band is shown for these choices of $g$ and particular values of $a$, $b$, and $c$.

\begin{remark}
Note that property (\ref{eqn:confidence-band}) is generally weaker than
\begin{equation*}
\lim_{M\to\infty} \PP \left(\{(t,F(t)):\, t>0\}\,\subset\,\tilde{B}_{\bgamma}^{(M)} \right) \geq q\,,
 \end{equation*}
 which corresponds to the common definition of an (asymptotic) confidence band for $F\in\mathcal{F}$. This restriction arises from the fact that the weak convergence in formula (\ref{eqn:hazardlimit}) only holds true on compact intervals,  while the limiting process $W_{\bgamma}$ is basically a Brownian motion on $[0,\infty)$.
 The weaker asymptotic formulation \eqref{eqn:confidence-band} is easier to verify and also sufficient for many applications, where failure times beyond a certain time horizon are not of practical relevance.
 \end{remark}

\begin{remark}
 Once having computed the quantile $d(q)$, it can also be used to state an asymptotic confidence band for the quantile function $F^{-1}$ of $F\in\mathcal{F}$. Inspecting the proof of Theorem \ref{Theorem: ConfBand} and letting $R=F^{-1}(p)$, we may rewrite
\begin{equation*}
    \sup_{t\in (0, R]} \frac{\sqrt{M}}{g(v_{\bgamma}(F(t)))}\frac{|\widehat{F}(t)-F(t)|}{1-F(t)} \,=\,\sup_{u\in(0,p]} \frac{\sqrt{M}}{g(v_{\bgamma}(u))}\, \frac{|\widehat{F}(F^{-1}(u))-u|}{1-u}\,,
    \end{equation*}
and, by introducing the random set
\begin{equation*}
\tilde{Q}_{\bgamma}^{(M,p)}\,=\,\left\{(u,z)\in(0,p]\times(0,\infty):\,\frac{\sqrt{M}}{g(v_{\bgamma}(u))}\,\frac{|\widehat{F}(z)-u|}{1-u}\,\leq d(q)\right\}\,,
\end{equation*}
it follows that 
 \begin{equation*}
     \forall\, p\in(0,1): \quad \lim_{M\to\infty} \PP \left(\{(u,F^{-1}(u)):\, u\in (0,p]\}\,\subset\,\tilde{Q}_{\bgamma}^{(M,p)} \right) \geq q\,.
 \end{equation*}
Since $\widehat{F}$ is a step function, it might occur for small sample sizes that the set $\{z:(u,z)\in\tilde{Q}_{\bgamma}^{(M,p)}\}$ is empty for some $u\in(0,p]$. This problem can be solved by either interpolating $\widehat{F}$ or by continuing the envelope as a step-function at empty cross sections.
\end{remark}

\begin{figure}
    \centering
    \begin{subfigure}{0.47\textwidth}
        \includegraphics[width=\textwidth]{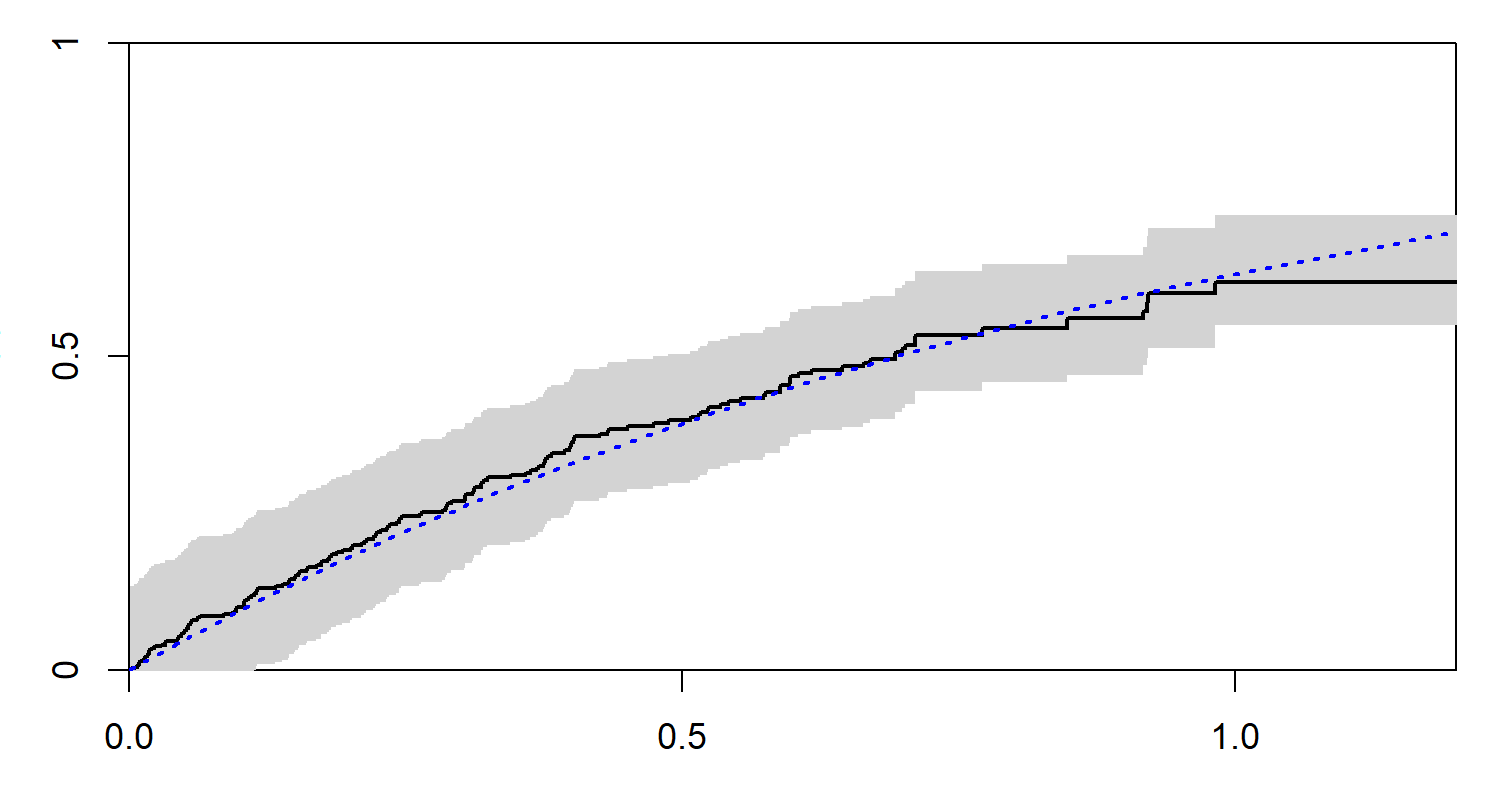}
        \caption{Known $\bgamma$; weighting $g(x)=x^{0.75}+0.5$.}
    \end{subfigure}
    \begin{subfigure}{0.47\textwidth}
        \includegraphics[width=\textwidth]{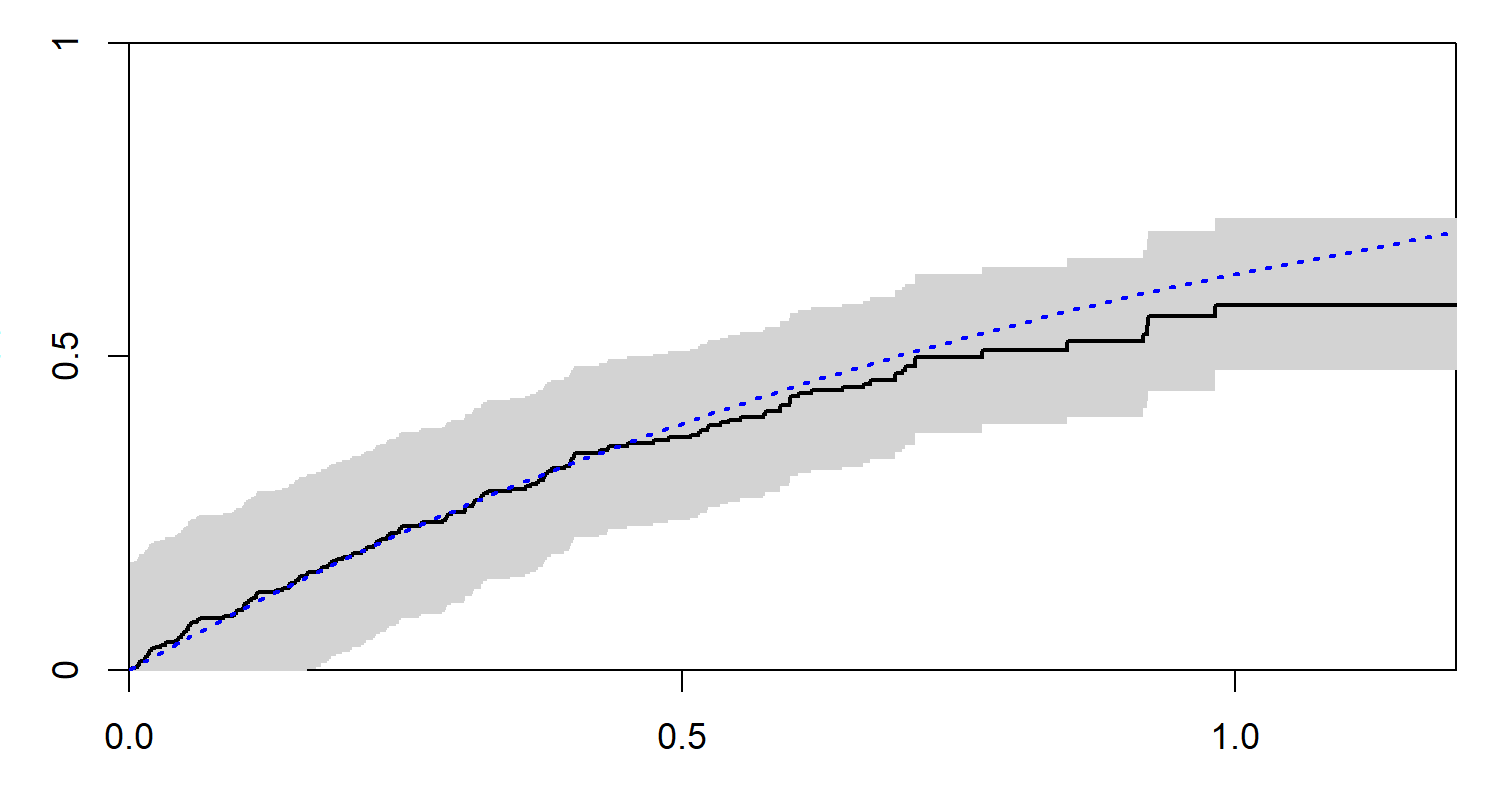}
        \caption{Unknown $\bgamma$; weighting $g(x)=x^{0.75}+0.5$.}
    \end{subfigure}
    \begin{subfigure}{0.47\textwidth}
        \includegraphics[width=\textwidth]{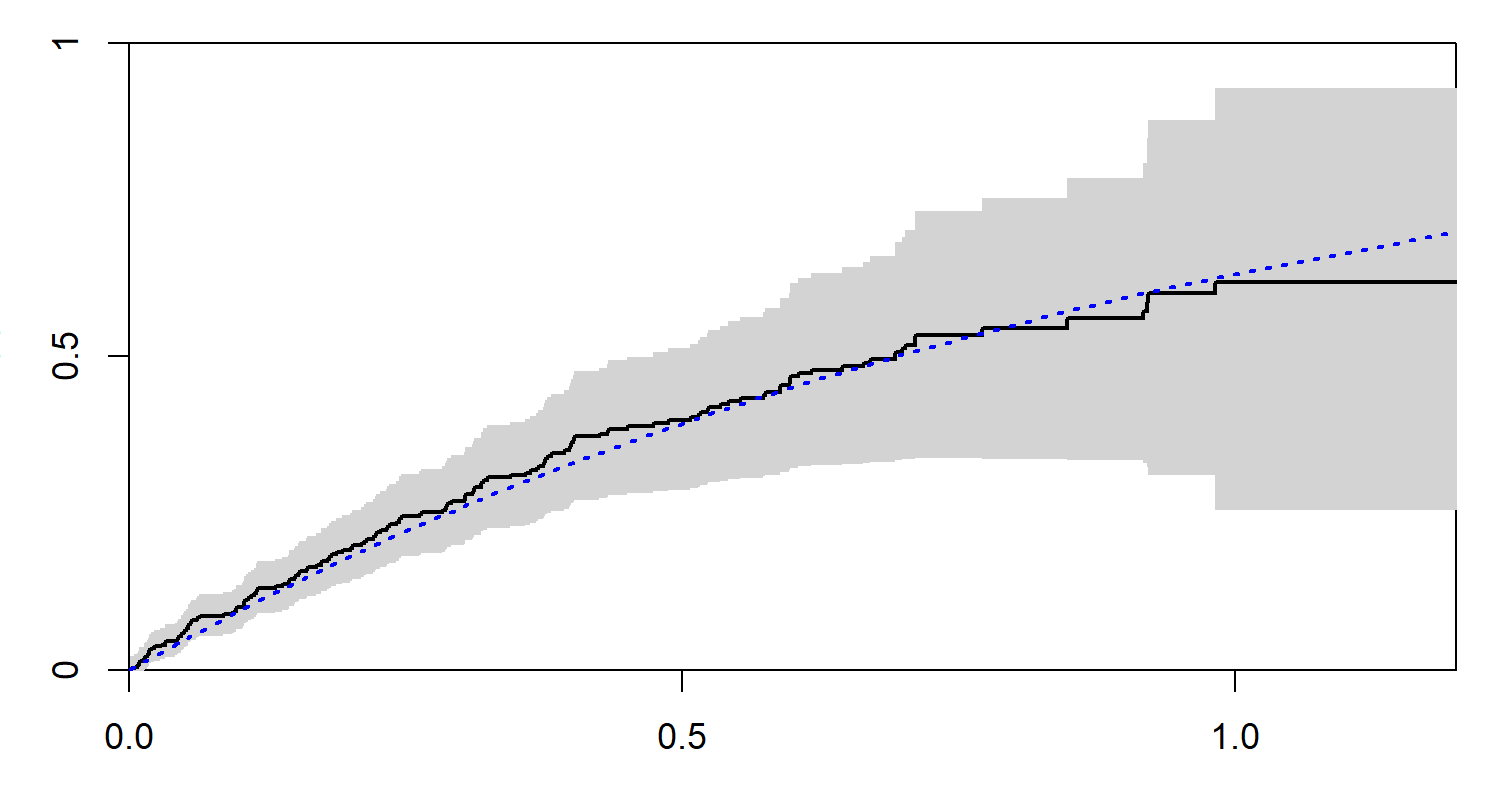}
        \caption{{\tabular[t]{@{}l@{}}Known $\bgamma$; \\weighting $g(x)=\sqrt{x \log(x+1)}+0.01$.\endtabular}}
    \end{subfigure}
    \begin{subfigure}{0.47\textwidth}
        \includegraphics[width=\textwidth]{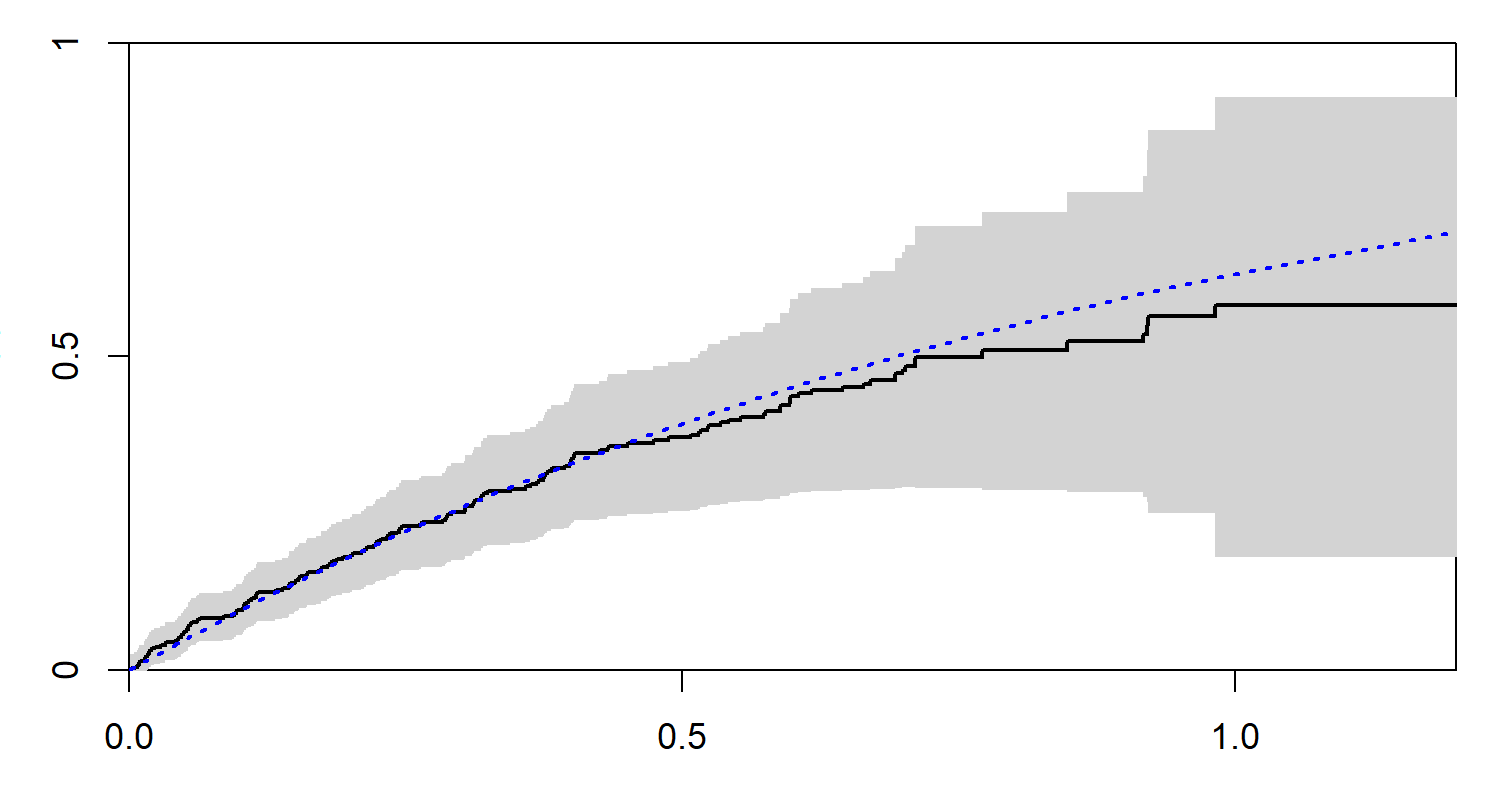}
        \caption{{\tabular[t]{@{}l@{}}Unknown $\bgamma$; \\weighting $g(x)=\sqrt{x \log(x+1)}+0.01$.\endtabular}}
    \end{subfigure}
    \caption{Baseline estimators (black) and asymptotic confidence bands (grey) based on $M=40$ sequential 4-out-of-10 systems with $\bgamma=(10,9,11,13)$ and standard exponential baseline cdf $F$ (blue dashed). The bands have asymptotic confidence level $90\%$.}
    \label{fig:example-bands}
\end{figure}

\subsection{Unknown Model Parameters}\label{ss:ABunknown}

The confidence band constructed in Theorem~\ref{Theorem: ConfBand} is infeasible if the load-sharing parameters $\gamma_1,\dots,\gamma_r$ are unknown, as both the estimator $\widehat{\Lambda}_F$ as well as the distribution of the statistics in formula \eqref{eqn:distr-free-band} depend on $\bgamma$. 
Of course, we can plug-in some estimator for $\bgamma$, but it does not seem to be possible to derive pivotal statistics for finite samples. Asymptotic confidence bands, however, can be constructed and are the subject matter in what follows.

Throughout this section we again assume that $\gamma_1=r$ to guarantee identifiability of $F$ and $\gamma_2,\dots,\gamma_r$. In that case, \cite{kvam2005estimating} suggest a profile-likelihood approach to estimate $\bgamma$ by maximizing the function
\begin{align}
    \ell_p(T;\bgamma) \,=\, \frac{1}{M}\sum_{i=1}^M \int_0^T \log \gamma^{(i)}(s) \, dN_i(s) - \int_0^T \log \overline{\gamma}(s)\, d\overline{N}(s)\,. \label{eqn:maxprofile}
\end{align}
The solution $\widehat{\bgamma}_T = \arg\max_{\bgamma} \ell_p(T;\bgamma)$ may then be used to estimate the cumulative hazard rate of $F$ in formula (\ref{eq:hatLa}) by
\begin{align*}
    \widehat{\Lambda}_F(t,\widehat{\bgamma}_T) \,=\, \int_0^t \frac{d\overline{N}(s)}{\widehat{\overline{\bgamma}}(s)}, 
        \quad \text{with}\quad 
    \widehat{\overline{\gamma}}(s) \,=\, \frac{1}{M} \sum_{i=1}^M \sum_{j=1}^r \widehat{\gamma}_j \mathds{1}\{N_i(s)=j-1\} \,.
\end{align*}

In Lemma \ref{la:concave}, some useful properties of $\ell_p$ are stated.

\begin{lemma}\label{la:concave}
    Any stationary point of $\ell_p$ is a global maximizer, and the maximizer is almost-surely uniquely determined. Moreover, $\ell_p$ is almost-surely strictly concave as a function of  $(\beta_2,\ldots, \beta_r)=(\log \gamma_2,\ldots, \log \gamma_r)$.
\end{lemma}

Maximizing formula \eqref{eqn:maxprofile} leads to the equivalent estimating equations
\begin{equation}\label{eqn:EE}
   U(T;\widehat{\bgamma}_T)_j  \,=\,0\,, \qquad j=2,\ldots, r\,,
   \end{equation}
   with
 \begin{equation*}
    U(t;\bgamma)_j \,=\, \gamma_j\cdot \frac{\partial}{\partial \gamma_j} \ell_p(t;\bgamma)\,=\, \frac{1}{M}\sum_{i=1}^M \int_0^t \left[ \delta_j^{(i)}(s) - \frac{\gamma_j\overline{\delta}_j(s)}{\overline{\gamma}(s)} \right] \,dN_i(s)\,,
    \end{equation*}
where
\begin{equation}\label{eq:deltaquer}
    \delta_j^{(i)}(s) \,=\, \mathds{1}\left\{N_i(s)=j-1\right\}\,,\;\; 1\leq i\leq M\,,
    \quad\text{and}\quad \overline{\delta}_j(s) = \frac{1}{M}\sum_{i=1}^M \delta_j^{(i)}(s)
\end{equation}
for $2\leq j\leq r$. Note that the integrand is upper bounded by $1$, and that $N_1,\dots,N_M$ have finitely many jumps, such that $U(\infty;\bgamma)$ is well-defined. In \cite{kvam2005estimating}, the asymptotic distribution of $\widehat{\bgamma}_T$ is derived  for $T<\infty$ with $\E(\gamma(T))>0$. However, this estimator does not use all the  information in the sample, as the cutoff at $T$ corresponds to a type-I right censoring of the data. This   censoring serves no statistical purpose, but only simplifies the asymptotic analysis of the estimator.
In the sequel, we extend their work by studying the estimator $\widehat{\bgamma}=\widehat{\bgamma}_\infty$ and by giving a refined description of the limit distribution. For this, we introduce the notation $\mathbf{diag}(d_1,\dots,d_k)$ for a diagonal $(k\times k)$-matrix with entries $d_1,\dots,d_k$ and $\mathcal{N}_k(\boldsymbol{a},\mathbf{A})$ for the $k$-dimensional normal distribution with mean vector $\boldsymbol{a}$ and covariance matrix $\mathbf{A}$.

\begin{thm}\label{thm:jointclt}
    For any $T^*\in (0,\infty]$,  the estimating equations \eqref{eqn:EE} have a unique solution $\widehat{\bgamma}_{T^*}$ with probability tending to one as $M\to \infty$, which is consistent, i.e.,  $\widehat{\bgamma}_{T^*}\to \bgamma$ in probability for $M\to \infty$.
    Moreover, for any $T\in (0,\infty)$ with $\PP (N(T)=r)<1$,  we have
    \begin{align*}
        \sqrt{M} \begin{pmatrix}
         \widehat{\bgamma}_{T^*} - \bgamma\\
            \left[\widehat{\Lambda}_F(t)-\Lambda_F(t)\right]_{t\in[0,T]} \\ 
            \left[\widehat{\Lambda}_F(t, \widehat{\bgamma}_{T^*}) - \Lambda_F(t)\right]_{t\in[0,T]}
        \end{pmatrix}
        \wconv \begin{pmatrix}
            \mathbf{diag}(\gamma_2,\ldots, \gamma_r)\, \overline{\boldsymbol{W}}_2(T^*) \\
            W_1(t)_{t\in [0,T]}  \\
            W_1(t)_{t\in [0,T]} + \boldsymbol{\Psi}(t;\bgamma) \mathbf{diag}(\gamma_2,\ldots, \gamma_r) \overline{\boldsymbol{W}}_2(T^*)
        \end{pmatrix},
    \end{align*}
    in the product space $\R^{r-1}\times D[0,T]\times D[0,T]$ as $M\rightarrow\infty$. 
    Here, the vector $\boldsymbol{\Psi}(t;\bgamma)$ has components 
    \begin{equation*}
    \Psi(t;\bgamma)_{j-1} = \int_0^t \frac{\E(\delta_j(s))}{\E (\gamma(s))}\,ds\qquad\text{with}\quad \delta_j(s)=\mathds{1}\left\{N(s)=j-1\right\}\quad\text{for $j=2,\ldots, r$\,,}
    \end{equation*}
the limit $W_1$ is a centered Gaussian process  with independent increments and variance function
    \begin{equation*}
        \Var(W_1(t))\,=\,\tau(t;\bgamma) \,=\, \int_0^t \frac{1}{\E (\gamma(s))}\, d\Lambda_F(s)\,,
        \end{equation*}   
    and $\overline{\boldsymbol{W}}_2(T^*)\sim\mathcal{N}_{r-1}(0, \mathbf{\Sigma}(T^*;\bgamma)^{-1})$ is a random vector independent of $W_1$, where the matrix $\mathbf{\Sigma}(t;\bgamma)$ has entries
    \begin{equation*}
        \Sigma(t;\bgamma)_{j-1,k-1}\,=\, \int_0^t \left[ \mathds{1}\{j=k\} \frac{\gamma_j \E (\delta_j(s))}{\E  (\gamma(s))} - \frac{\gamma_j \E (\delta_j(s))\cdot \gamma_k \E (\delta_k(s))}{\E (\gamma(s))^2} \right] \E (\gamma(s)) \,d\Lambda_F(s),
    \end{equation*}
    for $j,k=2,\ldots, r$ and $t\geq0$.     In particular, $\tau(T;\bgamma)$ and $\mathbf{\Sigma}(T^*;\bgamma)$ are finite.
\end{thm}

The benefit of Theorem \ref{thm:jointclt} is that it not only shows the asymptotic variance but also the asymptotic autocovariance structure of $\widehat{\Lambda}_F(t,\widehat{\bgamma}_{T^*})$, which improves upon previous results in the literature; see \cite{kvam2005estimating}, Theorem 2.
In particular, the limit process does not have independent increments due to the estimation error in $\gamma$.
A second, even more important distinction is that we allow for $T\neq T^*$ and $T^*=\infty$, whereas former statements require $T=T^*<\infty$. 
The latter restriction has been introduced for purely technical reasons to derive a functional central limit theorem for $\widehat{\Lambda}_F(t)$.
In statistical practice, however, we want to set $T^*=\infty$ and use all available data to estimate $\bgamma$. Theorem~\ref{thm:jointclt} now reveals that the joint limit theory holds in essentially the same way. 
Moreover, choosing $T^*=\infty$ yields the smallest asymptotic variance of $\widehat{\bgamma}_{T^*}$ and should therefore be considered the default ($\widehat{\bgamma} = \widehat{\bgamma}_\infty$). To allow for comparison with the literature, the formulation of Theorem \ref{thm:jointclt} also covers the suboptimal case $T^*<\infty$.

For the construction of asymptotic confidence bands for $F$ below, it is essential to choose $T^*$ independently from $T$. Here, our aim is to give a modified version of Theorem \ref{Theorem: ConfBand} that is applicable if $\bgamma$ is unknown. For this, let the quantities $\boldsymbol{\Psi}(u;\bgamma)$,  $\mathbf{\Sigma}(\infty;\bgamma)$, and $\tau(u;\bgamma)$ in Theorem \ref{thm:jointclt} for a standard uniform baseline cdf $G$ be denoted by $\boldsymbol{\Psi}_G(u;\bgamma)$, $\mathbf{\Sigma}_G(\bgamma)$, and $\tau_G(u;\bgamma)$$\;($$= v_{\bgamma} (u)$), which may be determined computationally by using the formulas
    \begin{eqnarray*}
        \E (\delta_j(u))  
        &=&\PP(N(u)=j-1) \\[1ex]
        \,&=&\,\left(\prod_{k=1}^{j-1} \gamma_k\right) \sum_{i=1}^{j} \left(\prod_{k=1,k\neq i}^j (\gamma_k-\gamma_i)\right)^{-1} (1-u)^{\gamma_i}\,,\qquad 2\leq j\leq r
        \end{eqnarray*}
        and
        \begin{equation*}
      \E (\gamma(u)) 
        = \sum_{i=1}^r (1-u)^{\gamma_i} \sum_{j=i}^r \left(\prod_{k=1}^j \gamma_k\right) \left(\prod_{k=1, k\neq i}^{j}(\gamma_k-\gamma_i)\right)^{-1}  
    \end{equation*}
    for $u\in(0,1)$; see Lemma \ref{Lemma: PN_t=r} and its proof. These formulas also extend continuously to the case that $\gamma_i=\gamma_k$ for some $i\neq k$, see Remark \ref{rem:egamma}.

\begin{thm}
	\label{Theorem:ConfBand2}
 Let $q\in (0,1)$ and the function $g:(0,\infty)\rightarrow\R$ be such that $g\geq\varepsilon$ for some $\varepsilon>0$ and $\sqrt{x\log\log x}/g(x)\rightarrow0$ and $g(x)\to\infty$ for $x\rightarrow\infty$. 
 Moreover, let
 \begin{align*}
    \overline{B}^{(M)}\,=\,\left\{(x,y)\in(0,\infty)\times(0,1):\,
    \frac{\sqrt{M}}{g(v_{\widehat{\bgamma}}(\widehat{F}(x)))}\frac{|\widehat{F}(x)-y|}{1-\widehat{F}(x)} \leq e_{\widehat{\bgamma}}(q) \right\} 
    \end{align*}
with $\widehat{F}(t)=1-\exp\{-\widehat\Lambda_F(t; \widehat{\bgamma})\}$, $\widehat{\bgamma}=\widehat{\bgamma}_{\infty}$, and $v_{\bgamma}$ as in formula (\ref{eq:tau}), where $e_{\bgamma}(q)$ denotes the $q$-quantile of the real-valued random variable
    \begin{equation}\label{eq:limitW*}
    \sup_{u\in (0,1)} \frac{\left|W_*(v_{\bgamma}(u)) + \boldsymbol{\Psi}_G(u;\bgamma) \mathbf{diag}(\gamma_2,\ldots, \gamma_r) \overline{W}\right|}{g(v_{\bgamma}(u))}
   \end{equation}
   for a standard Brownian motion $W_*$ and $\overline{W}\sim \mathcal{N}_{r-1}(0, \mathbf{\Sigma}_G(\bgamma)^{-1})$.
   Then $\overline{B}^{(M)}$ is a confidence band for the graph of $F\in\mathcal{F}$ with asymptotic level $q$ in the sense that
 \begin{equation}\label{eqn:confidence-band2}
     \forall\, R>0: \quad \lim_{M\to\infty} \PP \left(\{(t,F(t)):\, t\in (0, R]\}\,\subset\,\overline{B}^{(M)} \right) \geq q. 
 \end{equation}
\end{thm}

Application of the confidence band in Theorem \ref{Theorem:ConfBand2} to some data set $\mathbf{x}$, say, requires to compute the quantity $e_{\widehat{\bgamma}}(q)$.
Upon realizing the estimator $\widehat{\bgamma}$ based on the given data, $e_{\widehat{\bgamma}(\mathbf{x})}(q)$ can be approximated computationally via Monte Carlo simulations of the random variable in formula \eqref{eq:limitW*} for fixed  $\bgamma=\widehat{\bgamma}(\mathbf{x})$.
A closed-form expression for $e_{\bgamma}$ would of course be preferred, but seems to be analytically intractable to derive.

Finally, to illustrate the asymptotic confidence band in Theorem \ref{Theorem:ConfBand2}, we return to the data example at the end of Section \ref{ss:ABknown} and add the cases (b) and (d) in Figure \ref{fig:example-bands}, forming the counterparts to the cases (a) and (c) for known load-sharing parameters. A comparison of both shows that accounting for the sampling error in estimating $\bgamma$ does not alter the shape of the bands, but makes them slightly wider. 
This corresponds to the fact that the quantile $e_{\bgamma}(q)$ is larger than the quantile $d(q)$.

\section{Data Example}\label{s:data}

We illustrate the proposed confidence bands by means of two data examples from the literature. 

First, we consider a generated data supplied by \cite{kvam2005estimating}, which is supposed to describe pixel failures in plasma displays. 
Here, observations from $M=20$ independent 1-out-of-3 systems are obtained by simulations of the model estimated in \cite{bae_nonlinear_2004}.
The load-sharing parameters are unknown and estimated via the profile-likelihood method described in Section \ref{ss:ABunknown} with $T^*=\infty$ as $\widehat{\bgamma}_{T^*} = \widehat{\bgamma}_\infty = (3, 2.62, 1.25)$ resp.\ $\widehat{\balpha}_\infty = (1, 1.31, 1.25)$. 
Note that $\gamma_1=3$ resp.\ $\alpha_1=1$ is fixed by model assumption to allow for semiparametric identification of the baseline cdf $F$.
In contrast, the estimates reported by \cite{kvam2005estimating} are $\widehat{\balpha}_{KP}=(1, 1.64, 1.11)$. 
The difference may be explained by the fact that they need a smaller $T^*<\infty$ for their asymptotic theory, i.e., they artificially censor the data on the right, which is in fact not necessary as our extended theory reveals. 
Unfortunately, \cite{kvam2005estimating} do not state the value of $T^*$ used for their reported estimates.

As described in Section \ref{ss:ABunknown}, the estimated load-sharing parameter $\widehat{\bgamma}_\infty$ is used as plug-in to nonparametrically estimate the baseline cdf $F$, depicted as solid line in Figure \ref{fig:example-KP}. 
Our construction of confidence bands takes the statistical error in the estimated load-sharing parameters into account; see Theorem \ref{thm:jointclt}. 
For a nominal confidence level of $90\%$, the resulting confidence bands are shown in Figure \ref{fig:example-KP} for different weighting functions $g$, where 
the quantiles $e_{\bgamma}(0.9)$ are determined via 
$10^5$ simulations of the asymptotic distribution.
From the figure, it is evident that the choice of the weighting function affects the exact shape of the band, as every $g$  balances the type-I error differently across the real axis.

\begin{figure}
    \centering
    \begin{subfigure}{0.47\textwidth}
        \includegraphics[width=\textwidth]{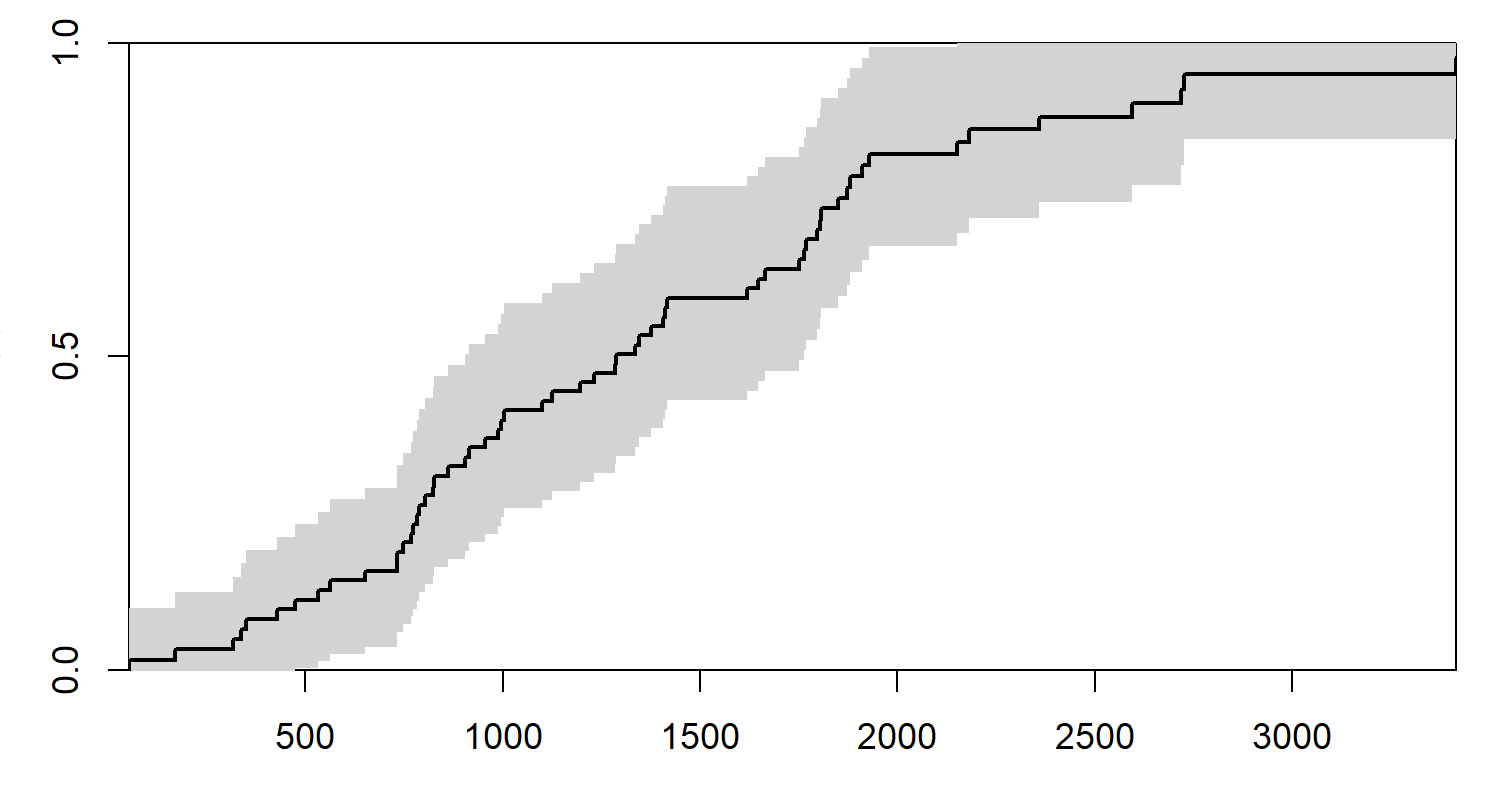}
        \caption{Weighting $g(x)=x^{0.7}+0.1$.}
    \end{subfigure}
    \begin{subfigure}{0.47\textwidth}
        \includegraphics[width=\textwidth]{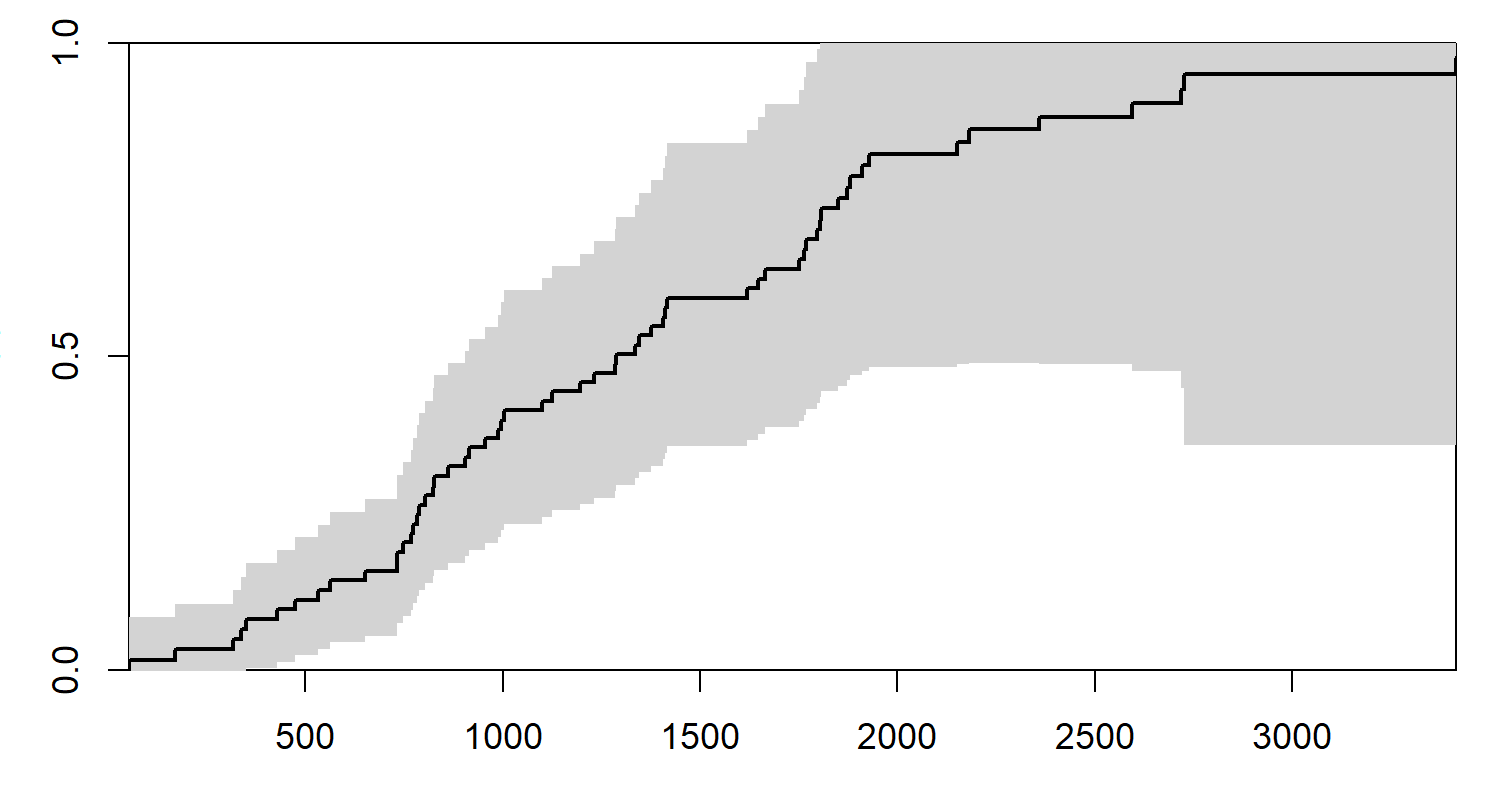}
        \caption{Weighting $g(x)=x+0.05$.}
    \end{subfigure}
    \begin{subfigure}{0.47\textwidth}
        \includegraphics[width=\textwidth]{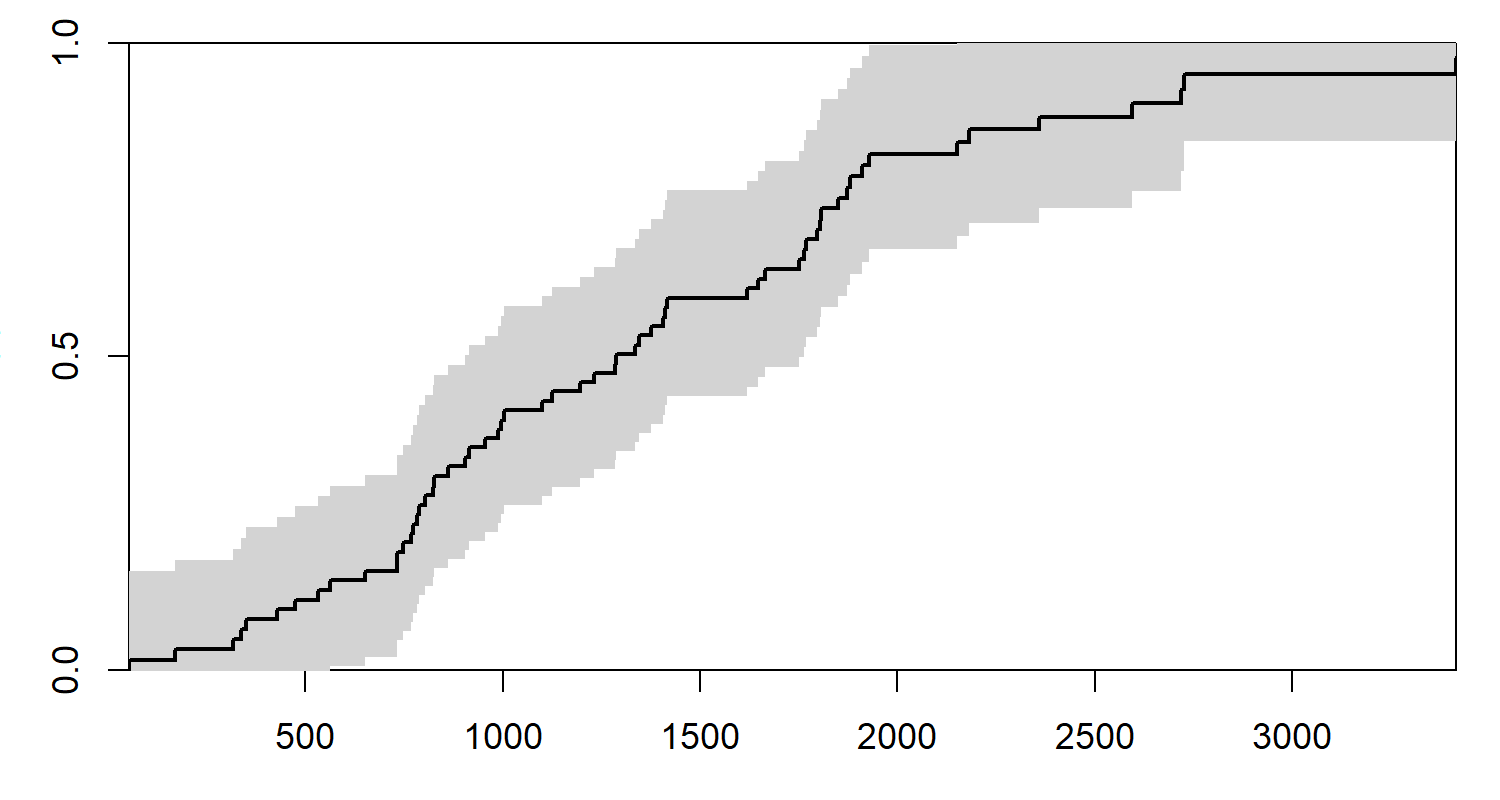}
        \caption{Weighting $g(x)=\sqrt{x \log(x+1)}+0.2$.}
    \end{subfigure}
    \caption{Estimated baseline cdf and confidence bands for $F$ with nominal confidence level 90\% for different weighting functions $g$ based on the data example in \cite{kvam2005estimating}  consisting of observations from 20 independent 1-out-of-3 systems.}
    \label{fig:example-KP}
\end{figure}

Secondly, we apply the same methodology to study a ReliaSoft data set \citep{Rel2002} which has previously been analyzed by \cite{sutar_accelerated_2014}, \cite{kong_cumulative-exposure-based_2016}, \cite{bedbur_inference_2019}, and \cite{MieBed2019}.
The data includes failure times of $M=18$ independent 1-out-of-2 systems with two continuously operating motors each, such that 36 failure times are observed in total. 
The load-sharing parameters are estimated as $\widehat{\bgamma}_\infty = (2, 2.51)$ resp.\ $\widehat{\balpha}_\infty=(1,2.51)$ indicating some load-sharing effect ($\alpha_2\neq1$), which was statistically confirmed by a test 
 with significance level 5\% in  \cite{MieBed2019}.
Using the statistical procedures developed in Section \ref{ss:ABunknown}, we are now able to construct semiparametric confidence bands for the baseline cdf, which are depicted in Figure \ref{fig:example-TM} for a nominal confidence level of 90\% and different weighting functions $g$.

\begin{figure}
    \centering
    \begin{subfigure}{0.47\textwidth}
        \includegraphics[width=\textwidth]{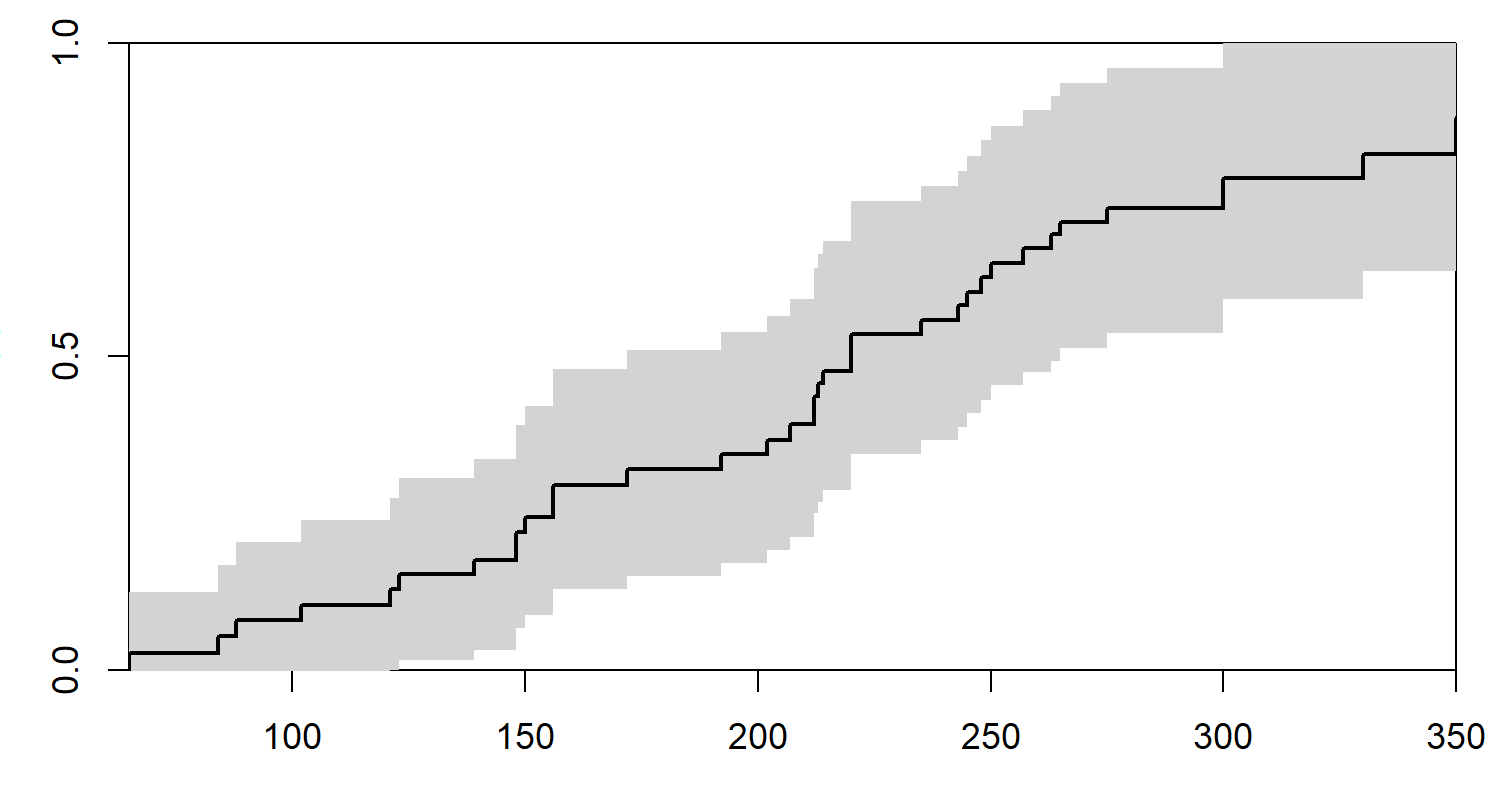}
        \caption{Weighting $g(x)=x^{0.7}+0.1$.}
    \end{subfigure}
    \begin{subfigure}{0.47\textwidth}
        \includegraphics[width=\textwidth]{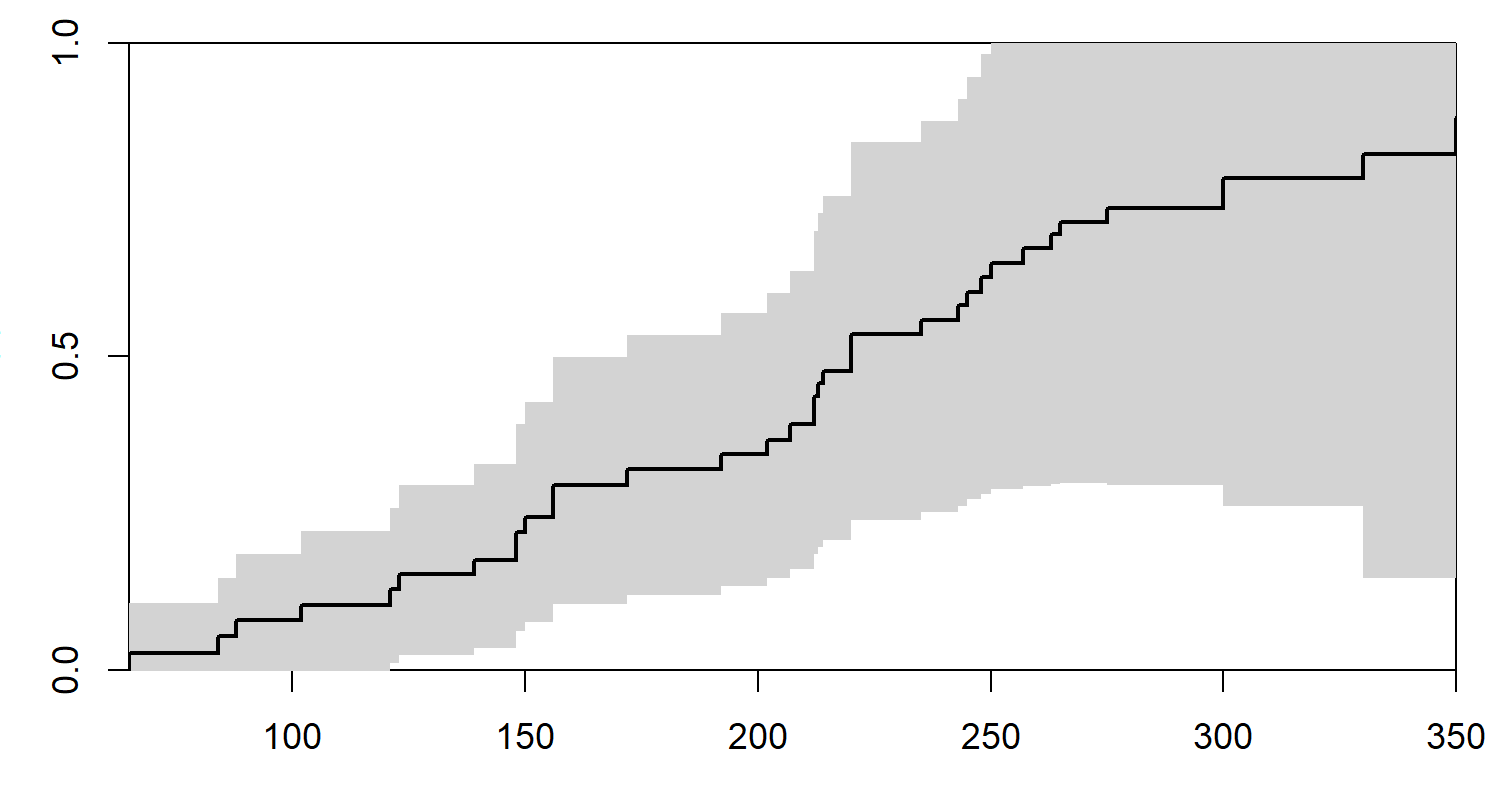}
        \caption{Weighting $g(x)=x+0.05$.}
    \end{subfigure}
    \begin{subfigure}{0.47\textwidth}
        \includegraphics[width=\textwidth]{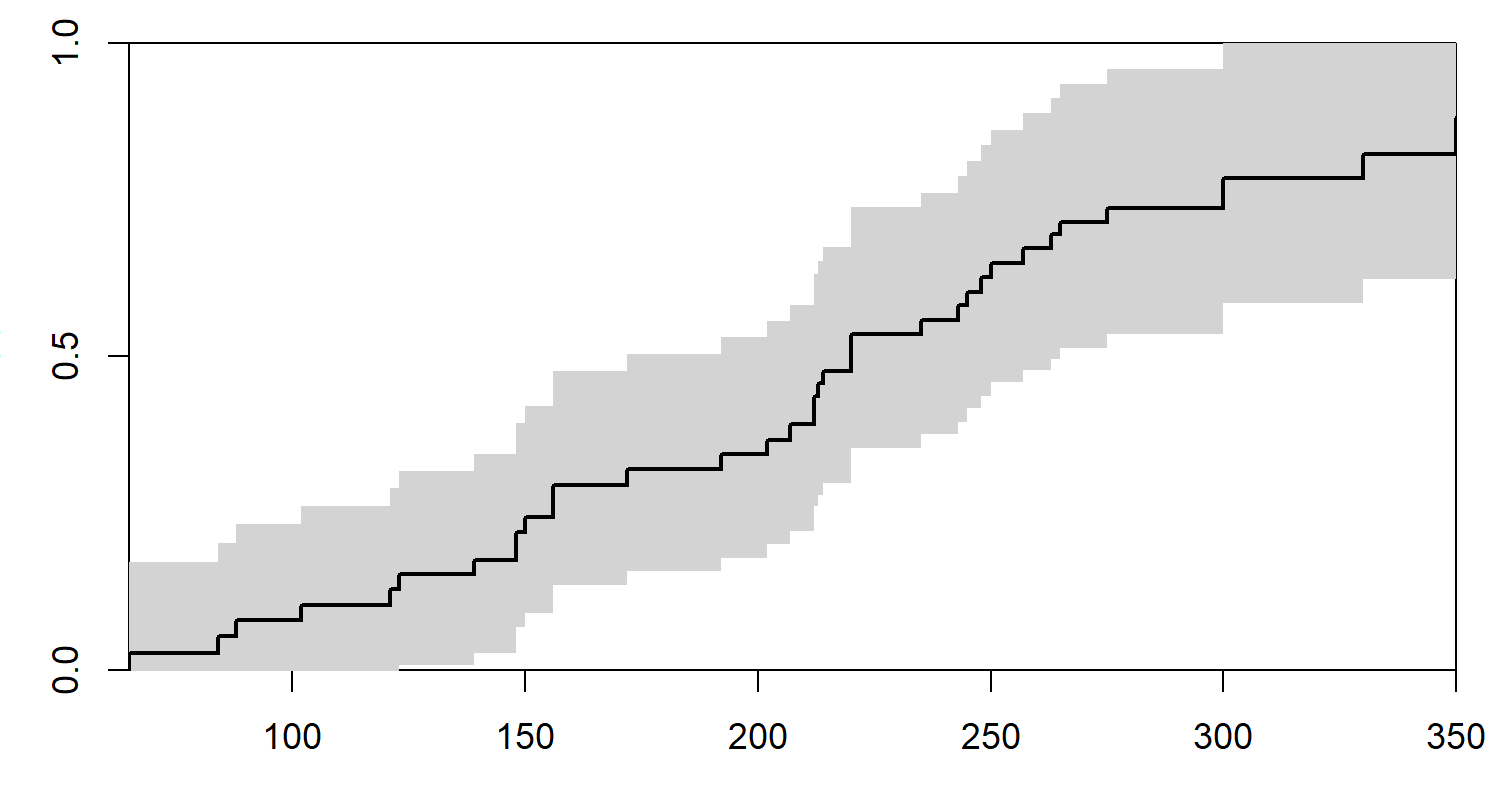}
        \caption{Weighting $g(x)=\sqrt{x \log(x+1)}+0.2$.}
    \end{subfigure}
    \caption{Estimated baseline cdf and confidence bands for $F$ with nominal confidence level 90\% for different weighting functions $g$ based on the two-motors data example in \cite{Rel2002}  consisting of observations from 18 independent 1-out-of-2 systems.}
    \label{fig:example-TM}
\end{figure}

\section{Conclusion}
    Confidence bands for the nonparametric baseline distribution function of a load-sharing system described by sequential order statistics can, in principle, be obtained by inverting exact finite sample tests.
    However, in case of unknown load-sharing parameters, the inversion is untractable, while asymptotic tests may be easily inverted.
    By introducing a suitable weight function, we can design the shape of the confidence band, effectively allocating the type-I error across the support of the baseline distribution.
    In particular, we may construct the bands to be tighter for small survival times. The rigorous treatment of the asymptotic confidence bands is enabled by a refined limit theory for the semiparametric estimators of the model, closing some current gaps in the literature. 

\section*{Appendix: Proofs and Auxiliary Results}
\renewcommand{\thesection}{A}

\begin{proof}[\underline{Proof of Theorem \ref{thm:EBknown}}]
     \noindent Applying formula (\ref{eq:NAErel}), we have 
\begin{eqnarray*}
    \PP\left(\text{gr}(F)\subset B_{\bgamma}\right)\,&=&\,\PP\left(\sup\limits_{t>0} H(\widehat{F}_{\bgamma}(t),F(t))\leq c_{\bgamma}(q)\right)\\[1ex]
    \,&=&\,\PP\left(\sup\limits_{t>0} H(\widehat{G}_{\bgamma}(F(t)),F(t))\leq c_{\bgamma}(q)\right)\\[1ex]
\,&=&\,\PP\left(\sup\limits_{u\in(0,1)} H(\widehat{G}_{\bgamma}(u),u)\leq c_{\bgamma}(q)\right)\\[1ex]
\,&=&\,q
    \end{eqnarray*}
    for every $F\in\mathcal{F}$.
\end{proof}


\begin{proof}[\underline{Proof of Lemma \ref{la:EN(t)}}]
\noindent According to \citet[Theorem 2.5]{CraKam2001b}, the marginal cdf of the $k^{\text{th}}$ GOSs $U^{(k)}$ based on a standard uniform  distribution and $\bgamma$ is given by
\begin{equation*}
F^{U^{(k)}}(t)\,=\,1-\left(\prod_{j=1}^k \gamma_j\right)\,\sum_{i=1}^k \frac{a_i^{(k)}}{\gamma_i}(1-u)^{\gamma_i}\,,\qquad u\in[0,1]\,,
\end{equation*}
where  the numbers $a_1^{(k)},\dots,a_r^{(k)}$ are defined by $a_i^{(k)}=\prod_{j=1,j\neq i}^k 1/(\gamma_j-\gamma_i)$ for $1\leq i\leq k$ and $1\leq k\leq r$. By using that $a_i^{(k)}=a_i^{(k+1)}(\gamma_{k+1}-\gamma_i)$ for $1\leq i\leq k$ and $1\leq k\leq r-1$, we obtain
 \begin{eqnarray*}
	\PP (N(u)=k)\,&=&\,\PP (N(u)\geq k) - \PP (N(u)\geq k+1)\\
	 \,&=&\, \PP (U^{(k)}\leq u)- \PP (U^{(k+1)}\leq u)\\[1ex]
	\,&=&\, \left(\prod_{i=1}^{k+1} \gamma_i \right) \sum_{i=1}^{k+1} \frac{a_i^{(k+1)}}{\gamma_i} (1-u)^{\gamma_i} - \left(\prod_{i=1}^{k} \gamma_i \right) \sum_{i=1}^{k} \frac{a_i^{(k)}}{\gamma_i} (1-u)^{\gamma_i}\\[1ex]
	\,&=&\,\left(\prod_{i=1}^{k} \gamma_i \right) \left[ \sum_{i=1}^k\frac{a_i^{(k+1)}\gamma_{k+1}-a_i^{(k)}}{\gamma_i}(1-u)^{\gamma_i} + \gamma_{k+1}\frac{a_{k+1}^{(k+1)}}{\gamma_{k+1}} (1-u)^{\gamma_{k+1}} \right]\\[1ex]
	\,&=&\,\left(\prod_{i=1}^{k} \gamma_i \right) \left[ \sum_{i=1}^k a_i^{(k+1)}(1-u)^{\gamma_i} + a_{k+1}^{(k+1)} (1-u)^{\gamma_{k+1}} \right]\\[1ex]
 \,&=&\, \left(\prod_{i=1}^k \gamma_i\right) \sum_{i=1}^{k+1} a_i^{(k+1)} (1-u)^{\gamma_i}\,,\qquad 1\leq k\leq r-1\,,
\end{eqnarray*}
and this representation is also valid for $k=0$ by the usual convention that empty products are equal to 1. By using formula (\ref{eq:gamma(s)}), this yields for $u\in[0,1]$
	\begin{eqnarray*}
\E (\gamma(u))\,&=&\,\sum_{k=1}^r \gamma_{k} \PP (N(u)=k-1)\
	\,=\,\sum_{k=1}^{r} \left(\prod_{j=1}^k \gamma_j\right) \sum_{i=1}^k a_i^{(k)} (1-u)^{\gamma_i}\\[1ex]
	\,&=&\,\sum_{i=1}^r (1-u)^{\gamma_i} \sum_{k=i}^r \left(\prod_{j=1}^k \gamma_j\right)a_i^{(k)}
	\,=\,\sum_{i=1}^r c_i(\bgamma)\,(1-u)^{\gamma_i}. 
\end{eqnarray*}
\end{proof}


\begin{proof}[\underline{Proof of Theorem \ref{Theorem: ConfBand}}]
\noindent
First, note that the assumptions on $g$ ensure that \\$\sup_{z\in(0,\infty)}|W_*(z)|/g(z)<\infty$ almost-surely, since the law of the iterative logarithm ensures that $\limsup_{z\rightarrow\infty}|W_*(z)|/\sqrt{2z\log\log z}=1$ almost-surely. Let $R>0$ and $p=F(R)\in(0,1)$. For $\widehat G=1-\exp\{-\widehat\Lambda_G\}$, formula  (\ref{eqn:hazardlimit}) along with the functional delta method yields that 
\begin{align}
\label{bConfBandProof1}
    \sqrt{M}\left(\widehat G(u)-u\right)\,\wconv\, (1-u)\,W_{\bgamma}(u)\,=\, (1-u)\,W_*(v_{\bgamma}(u))\,,
\end{align}
in the space $D[0,p]$ of cadlag functions on $[0,p]$, where $W_*$ denotes a standard Brownian motion. By using formula (\ref{eqn:shift}) and that $v_{\bgamma}$ is increasing, it follows that
\begin{eqnarray*}
\tilde{T}_{\bgamma,R}^{(M)}\,&=&\,\sup_{t\in (0, R]} \frac{\sqrt{M}}{g(v_{\bgamma}(\widehat{F}(t)))} \,\frac{| \widehat{F}(t)-F(t)|}{1-\widehat{F}(t)}	\\[1ex]
 \,&=&\, \sup_{u\in(0,p]} \frac{\sqrt{M}}{g(v_{\bgamma}(\widehat{G}(u)))} \,\frac{|\widehat{G}(u)-u|}{1-\widehat{G}(u)}\\[1ex]
	    \,&\wconv&\, \sup_{u\in (0,p]} \frac{|W_*(v_{\bgamma}(u))|}{g(v_{\bgamma}(u))}\,=\,
 \sup_{z\in (0,v_{\bgamma}(p)]} \frac{|W_*(z)|}{g(z)}\,,
\end{eqnarray*} 
the $q$-quantile of which will be denoted by $d_{\bgamma,R}(q)$. Here, the weak convergence is true by virtue of the continuous mapping theorem, since $\sqrt{M}(\widehat{G}(u)-u)\wconv (1-u) W_*(v_{\bgamma}(u))$ by formula (\ref{bConfBandProof1}) and $\widehat{G}(u)\wconv u$, and this convergence holds jointly because the second limit is not random. Finally, since $d_{\bgamma,R}(q)\leq d(q)$,
\begin{eqnarray*}
    \lim_{M\to\infty} \PP \left( \{(t, F(t)):\, t\in (0, R]\}\,\subset\,\tilde{B}_{\bgamma}^{(M)}  \right) 
    \,&=&\, \lim_{M\to\infty} \PP \left(\tilde{T}_{\bgamma,R}^{(M)}\leq d(q) \right) \\[1ex]
    \,&\geq&\,\lim_{M\to\infty} \PP \left(\tilde{T}_{\bgamma,R}^{(M)}\leq d_{\bgamma,R}(q) \right) 
    \,=\, q.    
\end{eqnarray*}
\end{proof}

\begin{proof}[\underline{Proof of Lemma \ref{la:concave}}]
    We first show that the right-hand side of equation (\ref{eqn:maxprofile}) is a concave function of $\boldsymbol{\beta}=(\beta_2,\dots,\beta_r)$, where $\beta_j = \log\gamma_j$, $1< j\leq r$ ($\beta_1=\log n$). Since $\log \gamma^{(i)}(s)$ is linear in $\boldsymbol{\beta}$ for $s\in[0,T]$ and $1\leq i\leq M$, the first summand in equation (\ref{eqn:maxprofile}) is linear in $\bbeta$ as well and thus also concave as a function of $\boldsymbol{\beta}$. It is therefore sufficient to show that $\log \overline{\gamma}(s)$ is a convex function of $\bbeta$ for $s\in[0,T]$, as this yields concavity in $\bbeta$ of the second summand in equation (\ref{eqn:maxprofile}). 
    To this end, note that $ \log\overline{\gamma}(s)$ can be represented as $h(\bbeta)=\log \sum_{j=1}^r Z_j(s) \exp(\beta_j)$ with nonnegative random variables $Z_j(s) = \sum_{i=1}^M \mathds{1}\{N_i(s)=j-1\}/M$ for $1\leq j\leq r$. The Hessian matrix of $h$ then has the entries 
    \begin{eqnarray*}
        \frac{\partial^2}{\partial \beta_i\, \partial \beta_k}h(\bbeta)
        \,&=&\, Y_k\,\mathds{1}\{i=k\}   -  Y_i Y_k \,,\qquad 2\leq i,k\leq r\,, 
    \end{eqnarray*}
    with
    \begin{equation*}
    Y_k \,=\,\frac{ Z_k \exp(\beta_k)}{\sum_{j=1}^r Z_j \exp(\beta_j)}\,\geq0\,,\qquad 1\leq k\leq r\,,
    \end{equation*}
    satisfying $\sum_{k=1}^r Y_k =1$, and it is positive semidefinite, because for $a_1=0$ and every vector $\boldsymbol{a}=(a_2,\dots,a_r)\in\R^{r-1}$
    \begin{equation*}
        \sum_{i,k=2}^r \left[\frac{\partial^2}{\partial \beta_i\, \partial \beta_k} h(\bbeta)\right]\,a_i\,a_k
        \,=\, \sum_{k=1}^r a_k^2 Y_k  - \left(\sum_{k=1}^r a_k Y_k\right)^2 \,\geq\,0
    \end{equation*}
    by using Jensen's inequality. Here, the inequality is strict unless $Y_k=1$ for some $k\in\{1,\dots,r\}$, which, in turn, is equivalent to $Z_j=0$ for all but one $j\in\{1,\dots,r\}$. Hence, the right-hand side of equation (\ref{eqn:maxprofile}) is strictly concave in $\bbeta$ unless at all jump times $s^*$ of $\overline{N}(s)$, $Z_j(s^*)=0$ for all but one $j\in\{1,\dots,r\}$ (all systems are in the same state). However, for $r\geq 2$, this event has probability zero, since almost-surely no two of the  counting processes $N_1,\dots,N_M$  jump at the same time due to the  continuity of $F\in\mathcal{F}$. Now, since the right-hand side of equation (\ref{eqn:maxprofile}) is almost-surely strictly concave as a function of $\bbeta$, it has at most one stationary point, which is a global maximum in case of existence. The same is then also true for the mapping $\bgamma\mapsto \ell_p(T;\bgamma)$.
\end{proof}

\begin{lemma}\label{lem:clt}
    For any $T\in(0,\infty)$ with $\PP (N(T)=r)<1$ and any $T^*\in(0,\infty]$, we have
    \begin{align*}
        \sqrt{M} \left[ \widehat{\Lambda}_F(t,\bgamma)_{t\in[0,T]},\, \bU(T^*;\bgamma) \right]\quad\wconv\quad \left(W_1(t)_{t\in [0,T]}, \, \boldsymbol{W}_2(T^*)\right)
    \end{align*}
    in the product space $D[0,T]\times \R^{r-1}$ as $M\rightarrow\infty$, where $W_1$ and $\boldsymbol{W}_2$ are independent centered Gaussian processes with independent increments in dimension $1$ resp. $r-1$, and
    \begin{equation*}
        \Var(W_1(t))\,=\,\tau(t) \,=\, \int_0^t \frac{1}{\E (\gamma(s))}\, d\Lambda_F(s)
        \end{equation*}
        and
        \begin{align*}
        \Cov(W_2(t))_{j-1,k-1}\,&=\,\Sigma(t)_{j-1,k-1}\,\\
        &=\, \int_0^t \left[ \mathds{1}\{j=k\} \frac{\gamma_j \E (\delta_j(s))}{\E  (\gamma(s))} - \frac{\gamma_j \E (\delta_j(s))\cdot \gamma_k \E (\delta_k(s))}{\E (\gamma(s))^2} \right] \E (\gamma(s)) \,d\Lambda_F(s),
    \end{align*}
    for $j,k=2,\ldots, r$ and $t>0$. In particular, $\tau(T)$ and $\mathbf{\Sigma}(T^*)=(\Sigma(T^*)_{j-1,k-1})$ are finite.
\end{lemma}

\begin{proof}[\underline{Proof of Lemma \ref{lem:clt}}]
    First, we show that, for any $t\in [0,\infty]$, $ \langle \sqrt{M}\, \bU(t;\bgamma) \rangle \to \mathbf{\Sigma}(t)$ in probability as $M\rightarrow\infty$. According to \cite{kvam2005estimating}, Lemma 1, we have the identity 
    \begin{equation*}
        U(t;\bgamma)_j 
        \,=\, \frac{1}{M}\sum_{i=1}^M \int_0^{t} \left[ \delta^{(i)}_j(s) - \frac{\gamma_j \overline{\delta}_j(s)}{\overline{\gamma}(s)} \right]\, \left( dN_i(s) - \gamma^{(i)}(s)\, d\Lambda_F(s) \right)\,.
    \end{equation*}
    This is a martingale, and its predictable quadratic variation is 
    \begin{equation*}
        \left\langle \sqrt{M}\, \bU(t;\bgamma) \right\rangle_{j-1,k-1} 
        \,=\, \int_0^t \left[ \mathds{1}\{j=k\} \frac{\gamma_j \overline{\delta}_j(s)}{\overline{\gamma}(s)} \;-\; \frac{\gamma_j \overline{\delta}_j(s)\cdot \gamma_k \overline{\delta}_k(s) }{\overline{\gamma}(s)^2} \right]\, \overline{\gamma}(s)\, d\Lambda_F(s),
    \end{equation*}
    for $j,k=2,\ldots,M$; see the proof of Lemma 2 in \cite{kvam2005estimating}. 
    Moreover, the authors prove that $\langle \sqrt{M}\, \bU(t;\bgamma)\rangle \to \mathbf{\Sigma}(t)$ in probability as $M\to\infty$, for any $t$ such that $\inf_{s\in [0,t]} \sum_{j=1}^r \gamma_j \PP (N(s)=j-1) >0$. 
    Since $\gamma_j>0$ for $1\leq j\leq r$, the latter condition is equivalent to $\PP(N(t)=r)<1$, and also equivalent to $\E (\gamma(t))>0$. However, it can be shown that this limit holds for all $t\in[0,\infty]$. To this end, note that by formulas (\ref{eq:hatLa}), (\ref{eq:gammaI}), and (\ref{eq:deltaquer}), we find
    \begin{equation*}
    \overline{\gamma}(s)
        \,=\, \frac{1}{M}\sum_{i=1}^M   \sum_{j=1}^r  \gamma_j\mathds{1}\{N_i(s) = j-1\} 
        \,=\,
     \sum_{j=1}^r \gamma_j \overline{\delta}_j(s) 
    \end{equation*}
   and thus
    \begin{align}
        &\E  \left| \left\langle \sqrt{M}\, \bU(\infty;\bgamma) \right\rangle_{j-1,k-1} - \left\langle \sqrt{M}\, \bU(t;\bgamma) \right\rangle_{j-1,k-1} \right| \nonumber \\
        \leq\, &\E  \int_t^\infty \left| \mathds{1}\{j=k\} \frac{\gamma_j \overline{\delta}_j(s)}{\overline{\gamma}(s)} \;-\; \frac{\gamma_j \overline{\delta}_j(s)\cdot \gamma_k \overline{\delta}_k(s) }{\overline{\gamma}(s)^2} \right|\, \overline{\gamma}(s)\, d\Lambda_F(s)\nonumber \\[1ex]
        \leq\, &\int_t^\infty \E (\overline{\gamma}(s))\, d\Lambda_F(s)\nonumber \\[1ex]
        =\, &\int_t^\infty \E (\gamma(s))\, d\Lambda_F(s) \nonumber \\[1ex]
        \leq\, &\left(\max_{1\leq j\leq r} \gamma_j\right)\int_t^\infty \PP \left(N(s)<r\right)\, d\Lambda_F(s)  \nonumber\\[1ex]
        =\, &\left(\max_{1\leq j\leq r} \gamma_j\right)\int_t^\infty \PP \left(\widetilde{N}(\Lambda_F(s))<r\right)\, d\Lambda_F(s) \nonumber\\[1ex]
        =\,  &\left(\max_{1\leq j\leq r} \gamma_j\right) \int_{\Lambda_F(t)}^\infty \PP \left(\widetilde{N}(z)<r\right)\, dz\,, \label{eqn:tailvar} 
    \end{align}
    where $\widetilde{N}$ is the counting process corresponding to GOSs with 
 a standard exponential baseline cdf and load-sharing parameter $\gamma_1,\dots,\gamma_r$.
    Note that the waiting times $L_1,\ldots, L_r$, say, between jumps of $\widetilde{N}$ are independent exponential random variables with rate parameters $\gamma_1,\ldots, \gamma_r$.
    Hence, we obtain
    \begin{align*}
        \PP \left(\widetilde{N}(z) <r\right) 
        \,\leq\, \sum_{j=1}^r \PP \left(L_j>\frac{z}{r}\right) 
        \,=\, \sum_{j=1}^r \exp\left(-\,\frac{\gamma_j z}{r}\right)\,,
    \end{align*}
    and the integral in formula \eqref{eqn:tailvar} is seen to be finite and vanishes for $t\to \infty$. This establishes the convergence $\langle \sqrt{M}\, \bU(t;\bgamma)\rangle \to \mathbf{\Sigma}(t)$ in probability as $M\rightarrow\infty$, for any $t\in [0,\infty]$.

    To obtain the limit of $\widehat{\Lambda}_F$, we  consider the decomposition
    \begin{equation*}
        \sqrt{M}\,(\widehat{\Lambda}_F(t,\bgamma)-\Lambda_F(t)) \,=\,A(t)-B(t)
        \end{equation*}
        with
        \begin{align*}
            &A(t)\,=\, \sqrt{M}\int_0^t \frac{\mathds{1}\{\overline{\gamma}(s)>0\}}{\overline{\gamma}(s)}\, \left(d\overline{N}(s) - \overline{\gamma}(s) \, d\Lambda_F(s)\right)
        \qquad\text{and}\\
        &B(t)\,=\,\sqrt{M}\int_{0}^t \mathds{1}\{\overline{\gamma}(s)=0\}\, d\Lambda_F(s)\,.
        \end{align*}
        We may bound the second term as 
    \begin{align*}
        \E (B(t)) &=  \sqrt{M} \int_0^t \PP \left(N_1(s)=r, \ldots, N_M(s)=r \right)\, d\Lambda_F(s) \\
        &\leq \sqrt{M} \,\Lambda_F(t)\, \PP (N(T)=r)^M.
    \end{align*}
    As $\PP (N(T)=r)<1$, the latter term tends to zero for $M\rightarrow\infty$.
    Regarding the local martingale $A(t)$, it is shown in \cite{kvam2005estimating}  that
    \begin{equation*}
        \langle A(t)\rangle \longrightarrow \int_0^t \frac{1}{\E (\gamma(s))}\, d\Lambda_F(s)\quad\text{in probability as $M\rightarrow\infty$}
        \end{equation*}
        and
        \begin{equation*}
        \langle A(t), U(t;\gamma)_{j-1}\rangle = 0, \qquad j=2,\ldots, r
    \end{equation*}
    ($A(t)$ coincides with the term $\Psi_3$, therein).    Since the jump sizes of $\widehat{\Lambda}_F$ and $U(t;\gamma)$ tend to zero, Rebolledo's martingale limit theorem (see, e.g., \cite{andersen2012statistical}, II.5.1) yields that 
    \begin{eqnarray*}
        &&\sqrt{M} \left( \widehat{\Lambda}_F(\min(t,T),\bgamma)-\Lambda_F(\min(t,T)),\; \bU(t;\gamma) \right)_{t\in [0, T^*]} \\[1ex]
        \wconv &&\left(W_1(\min(t,T)), \; \boldsymbol{W}_2(t)\right)_{t\in [0, T^*]}.
    \end{eqnarray*}
    This immediately proves the claim of the lemma for $T^*<\infty$.
    The case $T^*=\infty$ follows from formula \eqref{eqn:tailvar}, which implies $\limsup_{M\to \infty}\sqrt{M} |\bU(t;\gamma)-\bU(\infty;\gamma)|\to 0$ in probability as $t\to\infty$.
\end{proof}

\begin{lemma}\label{lem:uniformgradient}
    For any $T\in[0,\infty]$ and any compact set $K\subset \{r\}\times(0,\infty)^{r-1}$, it holds that
    \begin{equation*}
        \sup_{\bgamma^*\in K}\left| \frac{\partial}{\partial \gamma_k^*} \bU(T;\bgamma^*) - Z(T;\bgamma^*)_{j,k} \right| \longrightarrow0\quad\text{in probability as $M\rightarrow\infty$}
    \end{equation*}
    with limit
    \begin{equation*}
        Z(T;\bgamma^*)_{j,k} \,=\,  \int_0^T \left[ -\,\frac{\E (\delta_j(s))}{\E (\gamma^*(s))} \mathds{1}\{j=k\}  + \frac{\gamma^*_j \E (\delta_j(s)) \E (\delta_k(s))}{ \E (\gamma^*(s))^2} \right] \E (\gamma(s))\,d\Lambda_F(s),
    \end{equation*}
    for $j,k=2,\ldots, r$, and the matrix $\mathbf{Z}(T;\bgamma^*)=(Z(T;\bgamma^*)_{j,k})\in\R^{(r-1) \times (r-1)}$ is regular.
\end{lemma}

\begin{proof}[\underline{Proof of Lemma \ref{lem:uniformgradient}}]
     Setting $\overline{Z}(T;\bgamma^*)_{j,k}=\partial/\partial\gamma_k^*\,U(T;\bgamma^*)_j$ for $j,k=2,\dots,r$, we find 
    \begin{align*}
        \overline{Z}(T;\bgamma^*)_{j,k}\,=\, - \,\frac{\partial}{\partial \gamma^*_k}\int_0^T \frac{\gamma^*_j \overline{\delta}_j(s)}{\overline{\gamma}^*(s)}\, d\overline{N}(s) \,=\, \int_0^T \overline{\zeta}_{j,k}(s) \,d\overline{N}(s) 
    \end{align*}
    with
    \begin{equation*}
    \overline{\zeta}_{j,k}(s)\,=\, -\,\frac{\overline{\delta}_j(s)}{\overline{\gamma}^*(s)} \mathds{1}\{j=k\}  + \frac{\gamma^*_j \overline{\delta}_j(s)\overline{\delta}_k(s)}{ \overline{\gamma}^*(s)^2}\,.
    \end{equation*}
    Since the function $\bgamma^*\mapsto\overline{\mathbf{Z}}(T;\bgamma^*)$ is Lipschitz-continuous on $K$ uniformly in $M$, it suffices to establish the pointwise convergence $\overline{\mathbf{Z}}(T;\bgamma^*)\rightarrow \mathbf{Z}(T,\bgamma^*)$ as $M\rightarrow\infty$ for fixed $\bgamma^*$. 
    For this, let
    \begin{align*}
        \widetilde{Z}(T;\bgamma^*)_{j,k}
        \,=\, \int_0^T  \overline{\zeta}_{j,k}(s) \,\overline{\gamma}(s)\,d\Lambda_F(s)
    \end{align*}
    for $j,k=2,\dots,r$, which depends on the true parameter $\bgamma$ via $\overline{\gamma}(s)$. By the arguments in the proof of Lemma \ref{lem:clt}, we obtain $\E \|\langle \overline{\mathbf{Z}}(T;\bgamma^*) - \widetilde{\mathbf{Z}}(T;\bgamma^*) \rangle\|  \leq C/\sqrt{M}$ for some $C>0$, and thus $\E \|\overline{\mathbf{Z}}(T;\bgamma^*) - \widetilde{\mathbf{Z}}(T;\bgamma^*)\|^2 \to 0$ as $M\rightarrow\infty$.
 
    The uniform Glivenko-Cantelli theorem yields that for $j=1,\dots,r$
    \begin{align*}
        \frac{1}{M} \sum_{i=1}^M \mathds{1}\{N_i(s)\leq j-1 \} = \frac{1}{M} \sum_{i=1}^M \mathds{1}\{X_i^{(j-1)}\geq s \}\,\longrightarrow\, \PP ( X_i^{(j-1)}\geq s ) = \sum_{k=1}^j \E  (\delta_k(s))\,,
    \end{align*}
    in probability uniformly in $s$ and, hence, $\sup_{s\in [0,\infty)}|\overline{\delta}_j(s)-\E (\delta_j(s))|\to 0$ in probability as $M\rightarrow\infty$.
    This implies $\overline{\gamma}^*(s)\to \E (\gamma^*(s))$ and $\overline{\gamma}(s)\to \E (\gamma(s))$ uniformly in probability for $M\rightarrow\infty$.  
    Moreover, since $\overline{\zeta}_{j,k}(s)\leq \max_{1\leq i\leq r} (1/\gamma^*_i)$, the dominated convergence theorem yields
    $\widetilde{\mathbf{Z}}(T;\bgamma^*)\to \mathbf{Z}(T;\bgamma^*)$ as $M\rightarrow\infty$ for any $T\in[0,\infty)$.
    To draw the same conclusion for $T=\infty$, note that for any $j,k=2,\ldots, r$,
    \begin{align*}
        \E  \left| \widetilde{Z}(\infty;\bgamma^*)_{j,k}-\widetilde{Z}(T;\bgamma^*)_{j,k}\right| 
        &\leq  \left(\max_{1\leq i\leq r} 1/\gamma^*_i\right) \int_T^\infty \E (\gamma(s))\, d\Lambda_F(s)\,\rightarrow\,0\qquad\text{for $T\rightarrow\infty$}
    \end{align*}
    by using formula \eqref{eqn:tailvar}.
    Thus, we have shown that $\overline{\mathbf{Z}}(T;\bgamma^*)\to \mathbf{Z}(T;\bgamma^*)$ in probability as $M\rightarrow\infty$ for all $T\in[0,\infty]$ and all $\bgamma^*\in (0,\infty)^{r-1}$, which extends to uniform convergence in $\bgamma^*\in K$ via Lipschitz continuity.

    What is left to show is that $\mathbf{Z}(T;\bgamma^*)$ is a regular matrix. For this, let
\begin{equation*}
 r_j(s) \,=\, \frac{\gamma^*_j \E (\delta_j(s))}{\E (\gamma^*(s))}\,,\qquad j=1,\dots,r\,,
\end{equation*}
and note that $\sum_{j=1}^r r_j(s)=1$. As in the proof of Lemma \ref{la:concave}, it can be shown that the matrix $\mathbf{A}(s)\in\R^{(r-1)\times(r-1)}$ with entries
    \begin{equation*}
        A(s)_{j-1,k-1} 
        \,=\, -r_j(s) \mathds{1}\{j=k\} + r_j(s) r_k(s)\,,\qquad\,j,k=2,\dots,r \,,
           \end{equation*}
    is negative semidefinite for each $s$ such that $\E (\gamma^*(s))>0$.     Since $F$ has support $[0,\infty)$, it holds that $\E (\delta_j(s))>0$ and thus $r_j(s)>0$ for all $j=1,\ldots, r$.
    This implies that $\mathbf{A}(s)$ is strictly negative definite, and the same holds for the matrix
    \begin{equation*}
        \mathbf{Z}^\dagger(T;\bgamma^*)=\int_0^T A(s)\, \E (\gamma(s))\, d\Lambda_F(s),
    \end{equation*}
    which is thus regular. Therefore,
    \begin{equation*}
        \mathbf{Z}(T;\bgamma)=\mathbf{Z}^\dagger(T;\bgamma)\,\mathbf{diag}(1/\gamma_2^*,\dots,1/\gamma_r^*)
    \end{equation*}
    is also regular. 
\end{proof}
\begin{proof}[\underline{Proof of Theorem \ref{thm:jointclt}}]
    From Lemma \ref{lem:clt} and Lemma \ref{lem:uniformgradient}, the existence and consistency of $\widehat{\bgamma}_{T^*}$ follows from standard results on estimating equations; see, e.g., Appendix A in \cite{mies_estimation_2023}. 
    These results also yield the asymptotic normality of $\widehat{\bgamma}_{T^*}$ via the following technique: For some $\bzeta$ between $\widehat{\bgamma}_{T^*}$ and $\bgamma$, we have the Taylor expansion
    \begin{align*}
        \bU(T^*;\bgamma) 
        = \bU(T^*;\bgamma) - \bU(T^*;\widehat{\bgamma}_{T^*})
        = D\bU(T^*;\bzeta)\left(\bgamma-\widehat{\bgamma}_{T^*}\right),
    \end{align*}
    where $D$ denotes the Jacobian operator. 
    Consistency of $\widehat{\bgamma}_{T^*}$ implies that $\bzeta \rightarrow \bgamma$ in probability as $M\to \infty$.
    Moreover, by Lemma \ref{lem:uniformgradient}, we have that  $D\bU(T^*;\bgamma^*)\to \mathbf{Z}(T^*;\bgamma^*)$ locally uniformly in probability as $M\rightarrow\infty$, and the limit is regular.
    Thus, $D\bU(T^*;\bzeta)$ is invertible for sufficiently large values of $M$, and continuity arguments then yield that $D\bU(T^*;\bzeta)^{-1}\rightarrow \mathbf{Z}(T^*;\bgamma)^{-1}$ in probability and, thus,
    \begin{align}
     \sqrt{M}\left(\widehat{\bgamma}_{T^*} -\bgamma \right)
        \,=\, - D\bU(T^*;\bzeta)^{-1}\,\sqrt{M}\,\bU(T^*;\bgamma) 
       \; \wconv\; - \mathbf{Z}(T^*;\bgamma)^{-1} \boldsymbol{W}_2(T^*) \label{eqn:CLTgamma}
    \end{align}
    as $M\rightarrow\infty$.     Since $\mathbf{Z}(T^*;\bgamma)\, \mathbf{diag}(\gamma_2,\ldots,\gamma_r) = \mathbf{\Sigma}(T^*)$, we have 
    \begin{eqnarray*}
        &&\sqrt{M}\, \mathbf{diag}(\gamma_2,\ldots, \gamma_r)^{-1}\left(\widehat{\bgamma}_{T^*} -\bgamma \right)\\[1ex]
        \wconv \;&&-\mathbf{\Sigma}(T^*)^{-1} \boldsymbol{W}_2(T^*)=\overline{\boldsymbol{W}}_2(T^*) \;\sim\; \mathcal{N}_{r-1}\left(0, \mathbf{\Sigma}(T^*)^{-1}\right)\,.
    \end{eqnarray*}

    Now, we may follow the same steps as in the proof of Theorem 2 in \cite{kvam2005estimating} to find that
    \begin{align*}
        \sqrt{M}\left[\widehat{\Lambda}_F(t, \widehat{\bgamma}_{T^*}) - \Lambda_F(t)\right]
        &= \sqrt{M}\left[\widehat{\Lambda}_F(t, \bgamma) - \widehat{\Lambda}_F(t)\right] 
        \;-\;  \sqrt{M}\,\boldsymbol{\Psi}(t)  \, (\widehat{\bgamma}_{T^*}-\bgamma)  
        \;+\; \xi_m(t)\,,
    \end{align*}
    where $\sup_{t\in[0,T]}|\xi_m(t)|\to 0$ in probability as $M\to \infty$.
\end{proof}

\begin{proof}[\underline{Proof of Theorem \ref{Theorem:ConfBand2}}]
\noindent
Let $R>0$, $p=F(R)\in(0,1)$, and $G$ denote the standard uniform cdf. 
For $\widehat G=1-\exp\{-\widehat\Lambda_G\}$, Theorem~\ref{thm:jointclt} along with the functional delta method yields that 
\begin{equation*}
    \sqrt{M}\left(\widehat{G}(u;\widehat{\bgamma})-u\right)\quad\wconv\quad
    (1-u)\,\widetilde{W}(u;\bgamma)
    \end{equation*}
    in the space $D[0,p]$ of cadlag functions on $[0,p]$, where
    \begin{equation}
    \widetilde{W}(u;\bgamma)\,=\,
     W_*(v_{\bgamma}(u)) + \boldsymbol{\Psi}_G(u;\bgamma) \mathbf{diag}(\gamma_2,\ldots, \gamma_r) \mathbf{\Sigma}_G(\bgamma)^{-1/2} \overline{\mathbf{W}},
\end{equation}
for $\overline{\mathbf{W}} \sim \mathcal{N}_{r-1}(0, \mathbf{I}_{(r-1)\times (r-1)})$ following a standard multivariate normal distribution. Here, since $\tau_G(u;\bgamma)=v_{\bgamma}(u)$ by formula (\ref{eq:vgamma}), we have used that $W_1(u;\bgamma)=W_*(v_{\bgamma}(u))$ in distribution for a standard Brownian motion $W_*$.
In virtue of formula (\ref{eqn:shift}) and the continuous mapping theorem (cf. the proof of Theorem \ref{Theorem: ConfBand}), it follows that
\begin{eqnarray*}
	\tilde{T}_{\bgamma,R}^{(M)}\,&=&\,\sup_{t\in (0,R]} \frac{\sqrt{M}}{g(v_{\bgamma}(\widehat{F}(t;\widehat{\bgamma})))} \,\frac{| \widehat{F}(t;\widehat{\bgamma})-F(t)|}{1-\widehat{F}(t;\widehat{\bgamma})}	\\[1ex]
 \,&=&\, \sup_{u\in(0,p]} \frac{\sqrt{M}}{g(v_{\bgamma}(\widehat{G}(u; \widehat{\bgamma})))} \,\frac{|\widehat{G}(u; \widehat{\bgamma})-u|}{1-\widehat{G}(u; \widehat{\bgamma})}\\[1ex]
	    \,&\Rightarrow&\, \sup_{u\in (0,p]} \frac{|\widetilde{W}(u;\bgamma)|}{g(v_{\bgamma}(u))}
    \;=:\;\tilde{T}_{\bgamma,p}.
    \end{eqnarray*}
     Note that the assumptions on $g$ and the law of the iterated logarithm, i.e.,\\
     $\limsup_{t\to\infty} |W_*(t)|/\sqrt{t \log\log t} = \sqrt{2} $, imply that
    \begin{equation*}
    \tilde{T}_{\bgamma,p}\,\overset{p \to 1}{\longrightarrow}\, \tilde{T}_{\bgamma}\,=\,\sup_{u\in (0,1)} \frac{|\widetilde{W}(u;\bgamma)|}{g(v_{\bgamma}(u))}\,.
    \end{equation*}
Next, we show that $\bgamma\mapsto \tilde{T}_{\bgamma}$ is suitably continuous.
To this end, first note that the mappings $\bgamma\mapsto v_{\bgamma}(\cdot)$ and $\bgamma\mapsto \boldsymbol{\Psi}_G(\cdot;\bgamma)$ are continuous w.r.t.\ the supremum norm on $[0, p]$ and that $\bgamma\mapsto \mathbf{\Sigma}_G(\bgamma)^{-1/2}$ is continuous, such that, for any $p\in(0,1)$ and any fixed $\bgamma$,
\begin{align*}
     \sup_{\|\bgamma'-\bgamma\|\leq \delta} \left|\tilde{T}_{\bgamma,p}- \tilde{T}_{\bgamma',p}\right| \longrightarrow 0 \quad \text{in probability as } \delta \to 0.
\end{align*}
Moreover, by using that all components of $\boldsymbol{\Psi}_G(u,\bgamma)$ are bounded from above by \\
$1/\min_{1\leq j\leq r}\gamma_j$ and thus $\|\boldsymbol{\Psi}_G(u;\bgamma)\|\leq \sqrt{r-1}/\min_{1\leq j\leq r}\gamma_j$, we find
\begin{align*}
    \eta(p;\bgamma)
    =\,&\sup_{u\in(p,1)} \frac{|\widetilde{W}(u;\bgamma)|}{g(v_{\bgamma}(u))}\\
    \leq\,&\sup_{u\in(p,1)} \frac{|W_*(v_{\bgamma}(u))|}{g(v_{\bgamma}(u))} +  \frac{\sqrt{r-1}\max_{2\leq j \leq r} \gamma_j}{\bar{g}(v_{\bgamma}(p))\min_{1\leq j\leq r} \gamma_j} \|\mathbf{\Sigma}_G(\bgamma)^{-1/2}\|_2 \,\|\overline{W}\|\\[1ex]
    =\,&\sup_{t>v_{\bgamma}(p)} \frac{|W_*(t)|}{g(t)} +  \frac{C(\bgamma)}{\bar{g}(v_{\bgamma}(p))} \|\overline{W}\|
    \end{align*}
    with
    \begin{equation*}
    C(\bgamma) =\frac{\sqrt{r-1}\max_{2\leq j \leq r} \gamma_j}{\min_{1\leq j\leq r} \gamma_j} \|\mathbf{\Sigma}_G(\bgamma)^{-1/2}\|_2,
\end{equation*}
where $\bar{g}(p) = \inf_{y\geq p} g(y)$, $\|\mathbf{\Sigma}_G(\bgamma)^{-1/2}\|_2$ denotes the spectral norm of $\mathbf{\Sigma}_G^{-1/2}$, and $\bgamma\mapsto C(\bgamma)$ is continuous on $(0,\infty)^{r-1}$. For any compact set $K\subset (0,\infty)^{r-1}$, we have
\begin{align}
    \eta(p;K) 
    = \sup_{\bgamma \in K} \eta(p;\bgamma) 
    &\leq \sup_{t\geq v_K(p)} \frac{|W_*(t)|}{g(t)} +  \frac{C^*(K)\|\overline{W}\|}{ \bar{g}(v_{K}(p))}, \label{eq:eta}
\end{align}
for $C^*(K)=\sup_{\bgamma\in K} C(\bgamma)<\infty$, and $v_K(p)=\inf_{\bgamma \in K} v_{\bgamma}(p)$. 
 Moreover, by setting $\overline{c}(\bgamma)=\max_{1\leq i\leq r}c_i(\bgamma)>0$ and $\tilde{\gamma}=\min_{1\leq i\leq r}\gamma_i$, we obtain from formula (\ref{eq:tau}) and Lemma \ref{la:EN(t)} along with Remark \ref{rem:egamma} that 
\begin{eqnarray*}
   v_{\bgamma}(p) \,&=&\,
   \int_0^p\left[(1-s)\, \E (\gamma(s)) \right]^{-1}\,ds\,=\,\int_0^p\left[\sum_{i=1}^r c_i(\bgamma)\,(1-s)^{\gamma_i+1} \right]^{-1}\,ds\,\\[1ex]
   &\geq&\, \int_0^p\left[r\,\overline{c}(\bgamma)\,(1-s)^{\tilde{\gamma}+1}\right]^{-1}\,ds\,=\,\frac{(1-p)^{-\tilde{\gamma}}-1}{r\overline{c}(\bgamma)\tilde{\gamma}}\,\longrightarrow\,\infty\quad\text{for $p\rightarrow1$}\,.\
\end{eqnarray*}
In view of formula \eqref{eq:eta}, this yields that $\eta(p;K)\to 0$ for $p\to 1$, and thus
\begin{align*}
    \sup_{\bgamma\in K} \left|\tilde{T}_{\bgamma,p}-\tilde{T}_{\bgamma}\right| \,\leq\, \eta(p;K) \,\longrightarrow 0 \quad \text{in probability for $p\to 1$.}
\end{align*}
For any compact set $K$ with inner point $\bgamma$ and for sufficiently small $\delta>0$, we have
\begin{align*}
    \sup_{\|\bgamma'-\bgamma\|\leq \delta} \left| \tilde{T}_{\bgamma'} - \tilde{T}_{\bgamma} \right| \;\leq\;\sup_{\|\bgamma'-\bgamma\|\leq \delta} \left| \tilde{T}_{\bgamma',p} - \tilde{T}_{\bgamma,p} \right| + 2 \eta(p;K),
\end{align*}
and, hence, for any $\varepsilon>0$ and any $p\in(0,1)$,
\begin{align*}
    &\limsup_{\delta\to 0} \,\PP \left( \sup_{\|\bgamma'-\bgamma\|\leq \delta} \left| \tilde{T}_{\bgamma'} - \tilde{T}_{\bgamma} \right| > \varepsilon\right) \\
    \leq\,&\limsup_{\delta\to 0} \,\PP \left( \sup_{\|\bgamma'-\bgamma\|\leq \delta} \left| \tilde{T}_{\bgamma',p} - \tilde{T}_{\bgamma,p} \right| > \varepsilon/2\right)
    + \PP  \left( 2 \eta(p;K) > \varepsilon/2\right) \\
    =\,&\PP  \left(  \eta(p;K) > \varepsilon/4\right).
\end{align*}
Sending $p\to 1$ then gives
\begin{align}
    \sup_{\|\bgamma'-\bgamma\|\leq \delta} \left| \tilde{T}_{\bgamma'} - \tilde{T}_{\bgamma} \right| \,\longrightarrow\, 0 \quad \text{in probability as } \delta\to 0.\label{eq:continuity-T}
\end{align}
This continuity result now enables us to show that $\overline{B}^{(M)}$ has asymptotic level $q$ for $F$. For every $\varepsilon_1,\varepsilon_2>0$, we find
\begin{align*}
    &\lim_{M\to\infty} \PP \left( \{(t, F(t)):\, t\in (0, R]\}\,\subset\,\overline{B}^{(M)}  \right)\\
    =\,&\lim_{M\to\infty} \PP \left(\tilde{T}_{\bgamma,R}^{(M)}\leq e_{\widehat{\bgamma}}(q) \right) \\[1ex]
    \geq\,&\lim_{M\to\infty} \left[\PP \left(\tilde{T}_{\bgamma,R}^{(M)}\leq e_{\bgamma}(q-\varepsilon_1)-\varepsilon_2 \right) \,-\, \PP \left(  e_{\bgamma}(q-\varepsilon_1)-\varepsilon_2 \geq e_{\widehat{\bgamma}}(q) \right) \right]. 
\end{align*}
To show that the second probability vanishes as $M\to\infty$, we observe
\begin{align}
    &\quad e_{\bgamma}(q-\varepsilon_1)-\varepsilon_2 \,\geq\, e_{\bgamma'}(q) \nonumber \\
    \Longrightarrow& \quad \PP \left( \tilde{T}_{\bgamma'} \leq e_{\bgamma}(q-\varepsilon_1)-\varepsilon_2 \right) \geq q \nonumber \\
    \Longrightarrow & \quad \PP \left( \tilde{T}_{\bgamma} \leq e_{\bgamma}(q-\varepsilon_1)-\varepsilon_2/2 \right) + \PP \left( |\tilde{T}_{\bgamma}-\tilde{T}_{\bgamma'}| \geq \varepsilon_2/2 \right) \;\geq\; q \nonumber \\
    \Longrightarrow & \quad q-\varepsilon_1 + \PP \left( |\tilde{T}_{\bgamma}-\tilde{T}_{\bgamma'}| \geq \varepsilon_2/2 \right) \;\geq\; q. \label{eq:contradiction}
\end{align}
By virtue of  formula \eqref{eq:continuity-T}, there exists some $\delta(\varepsilon_1,\varepsilon_2)>0$ such that inequality \eqref{eq:contradiction} does not hold for $\|\bgamma-\bgamma'\|\leq \delta(\varepsilon_1,\varepsilon_2)$.
Hence,
\begin{align*}
   &\lim_{M\to\infty} \PP \left( \{(t, F(t)):\, t\in (0, R]\}\,\subset\,\overline{B}^{(M)}  \right)\\
    \geq\,&\lim_{M\to\infty} \left[\PP \left(\tilde{T}_{\bgamma,R}^{(M)}\leq e_{\bgamma}(q-\varepsilon_1)-\varepsilon_2 \right) \,-\, \PP \left(  \|\widehat{\bgamma}-\bgamma\|>\delta(\varepsilon_1,\varepsilon_2)\right) \right]\\
    \geq\,&\PP \left(\tilde{T}_{\bgamma}\leq e_{\bgamma}(q-\varepsilon_1) - \varepsilon_2 \right)\,.
\end{align*}
The distribution of $\tilde{T}_{\bgamma}$ is known to be absolutely continuous; see \cite{lifshits_absolute_1984}. 
Thus, since $\varepsilon_1,\varepsilon_2$ are arbitrary, we may let $\varepsilon_2\to 0$ and find that
\begin{eqnarray*}
    \lim_{M\to\infty} \PP \left( \{(t, F(t)):\, t\in (0, R]\}\,\subset\,\overline{B}^{(M)}  \right) 
    \,&\geq&\,\PP \left(\tilde{T}_{\bgamma}\leq e_{\bgamma}(q-\varepsilon_1)  \right) 
    \,=\, q-\varepsilon_1.
\end{eqnarray*}
Sending $\varepsilon_1\to 0$ completes the proof.
\end{proof}


\bibliography{Bib_Exp_Fam}
\bibliographystyle{apalike}

\end{document}